\documentclass{article}
\usepackage[a4paper]{geometry}
\usepackage[english]{babel}
\usepackage{amsthm}
\usepackage{amsmath}
\usepackage{a4wide}
\usepackage{amstext}
\usepackage{amssymb}
\usepackage{mathtools}
\usepackage{natbib}
\usepackage{hyperref}
\newcommand{\footremember}[2]{%
    \footnote{#2}
    \newcounter{#1}
    \setcounter{#1}{\value{footnote}}%
}

 \def\newblock{\ }%
 \bibpunct[, ]{[}{]}{,}{n}{}{,}%

\usepackage{multicol}
\usepackage{lipsum}
\def\R{\mathbb{R}}

\def\supp{\mathop{\rm supp}}

\def\Lip{\mathop{\rm Lip}}

\def\argmin{\mathop{\rm arg\, min}}

\def\P{{\mathcal P}}

\makeatletter
\newcommand{\proofpart}[2]{%
  \par
  \addvspace{\medskipamount}%
  \noindent\emph{Step #1: #2}\par\nobreak
  \addvspace{\smallskipamount}%
  \@afterheading
}
\makeatother
\providecommand{\keywords}[1]{\textbf{\textit{Keywords---}} #1}

\newtheorem{definition}{Definition}
\newtheorem{theorem}{Theorem}
\newtheorem{corollary}{Corollary}

\newtheorem{proposition}{Proposition}
\newtheorem{lemma}{Lemma}
\newtheorem{remark}{Remark}
\newtheorem{assumption}{Assumption}

\allowdisplaybreaks

\title{Robustness and Approximation of Discrete-time Mean-field Games under Discounted Cost Criterion}
\author{%
  U\u{g}ur Ayd{\i}n\footremember{alley}{Department of Mathematics, University of Connecticut, Storrs, CT, 06269, USA, ugur.aydin@uconn.edu}%
  \and Naci Saldi\footremember{trailer}{Department of Mathematics, Bilkent University, \c{C}ankaya, Ankara, 06800, TURKEY, naci.saldi@bilkent.edu.tr}%
  }
\begin{document}
\maketitle
\begin{abstract}
    In this paper, we investigate the robustness of stationary mean-field equilibria in the presence of model uncertainties, specifically focusing on infinite-horizon discounted cost functions. To achieve this, we initially establish convergence conditions for value iteration-based algorithms in mean-field games. Subsequently, utilizing these results, we demonstrate that the mean-field equilibrium obtained through this value iteration algorithm remains robust even in the face of system dynamics misspecifications. We then apply these robustness findings to the finite model approximation problem in mean-field games, showing that if the state space quantization is fine enough, the mean-field equilibrium for the finite model closely approximates the nominal one. 
\end{abstract}
\section{Introduction}
This paper deals with the robustness of mean-field games (MFGs) under infinite-horizon discounted cost function. In such game models, an individual agent interacts with a large population of other agents and competes against their collective behavior, which is represented by a mean-field term. This mean-field term is supposed to converge to the stationary distribution of the state of a typical single agent as the number of agents approaches infinity. In the limiting scenario, a typical agent faces a stochastic control problem similar to that of a single agent, but with a constraint on the state distribution at each time step. This constraint ensures that the state distribution aligns with the behavior of the entire population (i.e. mean-field term) when each agent applies the optimal policy against the mean-field term. This condition, which establishes stability between policy and state distribution, is known as the mean-field equilibrium.

The theory of mean-field games (MFGs) was introduced independently by Lasry and Lions in their paper \cite{LaLi07}, where they coined the term mean-field games, and by Huang, Malham'{e}, and Caines in their paper \cite{HuMaCa06}, where they approached it as stochastic dynamic games. These studies focus on continuous-time non-cooperative differential games that involve a large but finite number of agents, whose interaction becomes asymptotically negligible as their number increases. These works explore the infinite limits of these games to establish approximate Nash equilibria. In the realm of continuous-time differential games, the characterization of mean-field equilibrium is based on coupled Hamilton-Jacobi-Bellman (HJB) equations and Kolmogorov-Fokker-Planck (FPK) equations. For further exploration of continuous-time mean-field games with various models and cost functions, including games involving major-minor players, risk-sensitive games, games with Markov jump parameters, and Linear Quadratic Gaussian (LQG) games, we recommend referring to studies such as \cite{HuCaMa07,TeZhBa14,Hua10,BeFrPh13,Ca11,CaDe13,GoSa14,MoBa16}.

Compared to the extensive body of work in continuous-time frameworks, there is relatively less research available in the literature concerning discrete-time mean-field games. These studies primarily focus on scenarios where the state space is discrete, either finite or countable, and the agents' interactions are solely determined by their cost functions. In other words, the mean-field term does not impact the evolution of the agents' states. For instance, in \cite{GoMoSo10}, researchers investigate a mean-field game model with a finite state space. In \cite{AdJoWe15}, discrete-time mean-field games are examined under an infinite-horizon discounted cost criterion, and this study extends to unbounded state spaces.
Another line of research delves into discrete-time mean-field games with linear state dynamics, and this topic is explored in works such as \cite{ElLiNi13,MoBa15,MoBa16-cdc,NoNa13}. Additionally, references like \cite{Bis15,Sal19,WiAl05,Wie19} delve into discrete-time mean-field games while considering the average cost optimality criterion. In \cite{SaBaRaMOR2}, the authors investigate a discrete-time risk-sensitive mean-field game involving Polish state and action spaces. References \cite{SaBaRaSIAM,SaBaRaMOR1} tackle discrete-time mean-field games with Polish state and action spaces, exploring the discounted cost optimality criterion for both fully-observed and partially-observed scenarios.

The previously mentioned papers, excluding those focusing on linear models, have primarily addressed the existence of mean-field equilibrium without proposing algorithms that guarantee convergence to this equilibrium. In a recent paper \cite{anahtarci2020value}, authors tackled this issue in the context of mean-field games with abstract state and action spaces, considering both discounted cost and average cost criteria. They developed a value iteration algorithm and provided proof of its convergence to the mean-field equilibrium. Nevertheless, the robustness of discrete-time mean-field games, specifically when it comes to algorithms like value iteration used to calculate the mean-field equilibrium, remains relatively unexplored.

Robustness of mean-field games is mostly studied in the context of continuous-time mean-field games. In \cite{bauso2016robust}, the authors investigate mean field games that incorporate uncertainty in both states and payoffs. Their study revolves around a population of players, each characterized by individual states influenced by a standard Brownian motion and an additional disturbance term. The authors establish a mean field system tailored to these robust games. In \cite{MoBa16}, the authors investigate robust mean-field games in coupled Markov jump linear systems (MJLSs). Each of the $N$ agents operates under its own MJLS with distinct infinitesimal generators, influenced by both control inputs and individual disturbances. Utilizing robust mean field game theory, they design efficient decentralized controllers and analyze worst-case disturbance effects. By employing these controllers and their corresponding worst-case disturbances, which constitute a saddle-point solution to a generic stochastic differential game for MJLSs, the authors approximate the actual mean field behavior using a deterministic function. In \cite{MoBa17}, authors consider linear-quadratic risk-sensitive and robust mean-field games. In first problem, each agent minimizes an exponentiated cost function reflecting risk-sensitive behavior. In second problem, each agent minimizes a worst-case risk-neutral cost function due to an adversary's presence in their dynamics. They demonstrate that the mean-field equilibria in problem 1 and problem 2 exhibit a partial equivalence, with individual Nash strategies sharing the same control laws. However, there is a distinction in the mean-field terms of problem 1 and problem 2. Furthermore, they establish that as the parameter characterizing the robustness tends towards infinity, the two mean-field equilibria converge to become identical and equivalent to the risk-neutral case, similar to the principles seen in one-agent risk-sensitive and robust control theory.

For the approximation problem to numerically compute the mean-field equilibrium, there are various works available in the literature which are again mostly geared towards the continuous-time setup. In \cite{AcCa10}, authors propose discrete approximations by finite difference methods of the mean field system (coupled PDEs), both in the stationary case an non-stationary case. In their work \cite{AcCaCa12}, the authors focus on a specific class of mean-field games known as \emph{mean-field game planning problems} and introduce discrete approximations for these problems. Additionally, they explore the application of Newton's method to solve these discrete approximations. In \cite{NoRiDi17}, authors explore computational techniques for stationary mean-field games (MFG) and propose two algorithms. The first method uses gradient flow technique, which are derived from the variational characterization of specific MFG problems. The second one relies on monotonicity properties of MFG. We refer the reader to the lecture notes \cite{Lau21} for a more comprehensive overview of numerical methods for mean-field games. 

For finite horizon discrete time mean-field type control problems, model-based machine learning techniques have been employed to construct approximate models to the original mean-field control model in \cite{jusup2023safe,pasztor2023efficient} to approximate the optimal value of the original control model, which is the closest work to ours in terms of scope. However their techniques merely gives a policy and state measure that allows one to approximate the optimal value of the associated control problem rather than the mean-field equilibrium itself; whereas our work gives conditions on the approximations of the original model so that their associated mean-field equilibria can approximate both the state measure and optimal policy of the mean-field equilibrium of the original model. 

In the context of partially observed  Markov decision processes (MDPs), robustness problem for the optimal value of the control problem is also studied in \cite{baker2016continuity,kara2019robustness} under setwise and total variation convergence assumptions on transition probabilities due to the partially observed nature of the information structure. It is important to note that MFGs and MDPs have distinct objectives. While MFGs seek person-by-person optimal solutions, MDPs aim to achieve global optimal solutions. Consequently, our results cannot be directly compared to these existing works.

In terms of value iteration algorithms, using regularized (or softmax) policies instead of the bona-fide optimal policy gained traction in the literature recently due to their relaxed requirements \cite{anahtarci2023q, zaman2023oracle, cui2021approximately}. Due to their nature, the resulting mean-field policy in such case deviates from the original one, which can become uncontrollable under robust setting.

\subsection{Contributions}
In this paper, we explore the robustness of stationary mean-field equilibria when confronted with uncertainties in the model. Our approach involves first establishing conditions for the convergence of value iteration-based algorithms within the context of mean-field games. Subsequently, we leverage these results to illustrate that the mean-field equilibrium derived from this value iteration algorithm remains robust even in scenarios where the system dynamics are not precisely known. We then apply these insights into robustness to address the finite model approximation problem in mean-field games, demonstrating that finely quantized state spaces yield mean-field equilibria for the finite model that closely approximate the nominal solution. A comprehensive overview of the paper's contributions is provided below.

\begin{itemize}
\item[1.] In Lemma~\ref{lemma:b}, we obtain the Lipschitz coefficients of the components of the mean-field equilibrium operator, iterative application of which gives the mean-field equilibrium. Using this result, in Theorem~\ref{thrm1}, it is then shown that there exists an unique fixed point of the mean-field equilibrium operator, where this unique fixed point corresponds to a mean-field equilibrium. Moreover, it is proved that this unique fixed point can be computed by repeated application of mean-field equilibrium operator.
\item[2.] In Theorem~\ref{thrm}, we demonstrate that if we have MFGs with converging transition probabilities and assume that the value iteration algorithm works, and the nominal model has an unique equilibrium, then the mean-field equilibria obtained for the perturbed MFGs converge to those of the limiting case. To guarantee the convergence of the value iteration algorithm for both the perturbed MFGs and the nominal model, it is crucial to rely on the outcome described in item 1. This is the reason behind our comprehensive establishment of the value iteration algorithm's convergence.
\item[3.] In Section~\ref{q-mfgs}, we construct the quantized mean-field games. Then, in Lemma~\ref{lemma:c}, we establish the Lipschitz continuity properties of the components of the quantized games, which enables us to ensure the convergence of the value iteration algorithm for the quantized games (see Lemma~\ref{lemma:e}). In Section~\ref{e-mfgs}, we extend the quantized games to the original state space and prove in Proposition~\ref{prop:d} the convergence of the value iteration algorithm for the extended games. Finally, in Theorem~\ref{thrm4}, we establish that mean-field equilibria for the extended games converge to the mean-field equilibrium of the nominal model. 
\item[4.] In Section~\ref{g-robust}, we improve Theorem \ref{thrm} and achieve the most comprehensive robustness result in mean-field games. Nevertheless, the applicability of this result might be limited unless specific conditions are established to guarantee the convergence of the value iteration algorithm for both the approximating and nominal models.
\end{itemize}

\section{Mean-field Games}

A discrete-time mean-field game is specified by
$
\bigl( X, A, p, c, \mu_0 \bigr), \nonumber
$
where $X$ and $A$ are the state and the action spaces, respectively. Here, $X$ is a Polish space (complete separable metric space) with the metric $d_{X}$ and $A$ is a compact and convex subset of a finite dimensional Euclidean space $\R^d$ with the Euclidean norm $\|\cdot\|$. The measurable function $p : X \times A \times \P(X) \to \P(X)$ denotes the transition probability of the next state given the previous state-action pair and the state-measure. The measurable function $c: X \times A \times \P(X) \rightarrow [0,\infty)$ is the one-stage cost function. The measure $\mu_0$ is the initial state distribution. 

We introduce the history spaces $H_0 = X$ and $H_{t}=(X\times A)^{t}\times X$ for $t=1,2,\ldots$, all endowed with product Borel $\sigma$-algebras. In this model, a policy is a sequence $\varphi=\{\pi_{t}\}$ of stochastic kernels on $A$ given $H_{t}$. In this paper, we only consider stationary policies. A stationary policy $\pi$ is a stochastic kernel on $A$ given $X$; that is, $\pi:X \rightarrow \P(A)$ is a measurable function. Let $\Pi$ denote the set of all  stationary policies. By the Ionescu Tulcea Theorem \cite{HeLa96}, a stationary policy $\pi$ and an initial measure $\mu_0$ define a unique probability measure $P^{\pi}$ on $(X \times A)^{\infty}$. The expectation with respect to $P^{\pi}$ is denoted by $E^{\pi}$.

Let us fix a state-measure $\mu \in \P(X)$ that describes the collective behavior of the other agents. 
A stationary policy $\pi^{*} \in \Pi$ is said to be optimal for $\mu$ if
$
J_{\mu}(\pi^{*},\mu_0) = \inf_{\pi \in \Pi} J_{\mu}(\pi,\mu_0) \nonumber
$
where 
\begin{align}
J_{\mu}(\pi,\mu_0) &= E^{\pi}\biggl[ \sum_{t=0}^{\infty} \beta^t c(x(t),a(t),\mu) \bigg| x(0) \sim \mu_0 \biggr] \nonumber  
\end{align}
is the discounted cost of stationary policy $\pi$ under the state-measure $\mu$. Here, $\beta \in (0,1)$ is the discount factor. In this model, the evolution of the states and actions is given by
\begin{align}
x(0) \sim \mu_0, \,\,\,
x(t) \sim p(\,\cdot\,|x(t-1),a(t-1),\mu), \text{ } t\geq1, \,\,\,
a(t) \sim \pi(\,\cdot\,|x(t)), \text{ } t\geq0. \nonumber
\end{align}
Now, define the set-valued mapping $\Psi : \P(X) \rightarrow 2^{\Pi}$  as 
$\Psi(\mu) = \{\pi \in \Pi: \pi \text{ is optimal for }  \mu \text{ }\text{ and } \text{ } \mu_0 = \mu\};$
that is, given $\mu$, the set $\Psi(\mu)$ is the set of optimal stationary policies for $\mu$ when the initial distribution is $\mu$ as well.  

Conversely, we define another set-valued mapping $\Lambda : \Pi \to 2^{\P(X)}$ as follows: given $\pi \in \Pi$, the state-measure $\mu_{\pi}$ is in $\Lambda(\pi)$ if it is a fixed point of the equation
\begin{align}
\mu_{\pi}(\,\cdot\,) = \int_{X \times A} p(\,\cdot\,|x,a,\mu_{\pi})  \, \pi(da|x) \, \mu_{\pi}(dx). \nonumber
\end{align}
In other words, $\mu_{\pi}$ is the invariant distribution of the Markov transition probability $P(\cdot|x) = \int_{A} p(\,\cdot\,|x,a,\mu_{\pi}) \, \pi(da|x)$. 

The notion of equilibrium for mean-field games is defined via these mappings $\Psi$, $\Lambda$ as follows.

\begin{definition}
A pair $(\pi_*,\mu_*) \in \Pi \times \P(X)$ is a stationary \emph{mean-field equilibrium} if $\pi_* \in \Psi(\mu_*)$ and $\mu_* \in \Lambda(\pi_*)$. In other words, $\pi_*$ is an optimal  stationary policy given the state-measure $\mu_*$ and $\mu_*$ is the state distribution under the policy $\pi_*$. 
\end{definition}

In the literature, the existence of mean-field equilibria has been established for the discounted cost in \cite{SaBaRaSIAM}. 

\section{Fixed Points of the Value Iteration for Mean-field Games}

For any measurable function $u: X \times A \rightarrow \R$, let $u_{\min}(x) \coloneqq \inf_{a \in A} u(x,a)$ and $u_{\max}(x) \coloneqq \sup_{a \in A} u(x,a)$. Let $w:X \times A \rightarrow [1,\infty)$ be a continuous weight function. For any measurable $v: X \times A \rightarrow \R$, we define $w$-norm of $v$ as
$$
\|v\|_w \coloneqq \sup_{(x,a) \in X \times A} \frac{|v(x,a)|}{w(x,a)}. 
$$ 
For any measurable $u: X \rightarrow \R$, we define $w_{\max}$-norm of $u$ as
\begin{align}
\|v\|_{w_{\max}} &\coloneqq \sup_{x \in X} \frac{|u(x)|}{w_{\max}(x)}. \nonumber
\end{align}
Let $B(X,K)$ be the set of real-valued measurable functions with $w_{\max}$-norm less than $K$. Let $C(X)$ be the set of real-valued continuous functions on $X$ and let $C_{w_{\max}}(X)$ denote the set of elements in $C(X)$ with finite $w_{\max}$-norm. For each $g \in C(X)$, let
\begin{align}
\|g\|_{\Lip} \coloneqq \sup_{\substack{(x,y)\in X\times X \\ x \neq y}} \frac{|g(x)-g(y)|}{d_{X}(x,y)}. \nonumber
\end{align}
If $\|g\|_{\Lip}$ is finite, then $g$ is called Lipschitz continuous with Lipschitz constant $\|g\|_{\Lip}$. $\Lip(X)$ denotes the set of all Lipschitz continuous functions on $X$, i.e.,
$
\Lip(X) \coloneqq \{g \in C(X): \|g\|_{\Lip} < \infty \} 
$
and $\Lip(X,K)$ denotes the set of all $g \in \Lip(X)$ with $\|g\|_{\Lip} \leq K$. For any $\mu, \nu \in \P(X)$, we denote by $C(\mu,\nu)$ the set of couplings between $\mu$ and $\nu$; that is, $\xi \in \P(X\times X)$ is an element of $C(\mu,\nu)$ if $\xi(\cdot \times X) = \mu(\cdot)$ and $\xi(X\times\cdot) = \nu(\cdot)$. The \emph{Wasserstein distance of order $1$} \cite[Definition 6.1]{Vil09} between two probability measures $\mu$ and $\nu$ over $X$ is defined as
$$
W_1(\mu,\nu) = \inf\left\{ \int_{X \times X} d_{X}(x,y) \, \xi(dx,dy): \xi \in C(\mu,\nu) \right\}. 
$$
By using Kantorovich-Rubinstein duality, we can also write Wasserstein distance of order $1$ \cite[p. 95]{Vil09} as follows:
\begin{align}
W_1(\mu,\nu) \coloneqq \sup \biggl\{\biggl|\int_{X} g d\mu - \int_{X} g d\nu\biggr|: g \in \Lip(X,1)\biggr\}. \nonumber
\end{align}
For compact $X$, the Wasserstein distance of order $1$ metrizes the weak topology on the set of probability measures $\P(X)$ (see \cite[Corollary 6.13, p. 97]{Vil09}). However, in general, it is stronger than weak topology. 

For each fixed $\gamma \in \mathbb R_{>0}$, we set a  domain for $Q$-functions as follows
$$ \mathcal C_{\gamma} = \bigg \{ Q:X \times A \rightarrow [0,\infty): \| Q \|_{w} \leq \frac{\gamma M}{1-\beta \alpha}, \| Q_{\min} \|_{\Lip} \leq \frac{\gamma L_2}{1-\beta K_2} \bigg \}$$
where constants $M$, $\alpha$, $L_2$, and $K_2$ are defined below. We define mean-field equilibrium operator as $H_{\gamma}: \mathcal C_{\gamma} \times \mathcal P(X) \rightarrow \mathcal C_{\gamma} \times \mathcal P(X) :(Q,\mu) \rightarrow (H^{\gamma}_1(Q,\mu),H^{\gamma}_2(Q,\mu))$, where the iterations of the $Q$-functions is handled by the operator
$$H^{\gamma }_1(Q,\mu)(x,a) = \gamma c(x,a,\mu) +\beta \int _X Q_{\min}(y)p(dy|x,a,\mu)$$ and by using the function  
$f_{\gamma }(x,Q,\mu)= \underset{a \in A}{\mathrm{argmin}} H^{\gamma}_1(Q,\mu)(x,a)$, we define the following ergodicity operator to handle the iterations on the state-measures
$$H^{\gamma}_2(Q,\mu)(\cdot) = \int_{X} p(\cdot|x, f_{\gamma}(x,Q,\mu),\mu)\mu(dx).$$

Following \cite[Assumption 1]{anahtarci2020value}, we will work under  the following set of assumptions for the operator $H_{1}$ to obtain its fixed point:

\begin{assumption}
\label{assump1}
\begin{itemize}
\item [ ]
\item[a)] The cost function $c$ is continuous and the transition probability $p$ is weakly continuous. Moreover, we have the following bounds:
\begin{align}
& c(x,a,\mu)\leq Mw(x,a)
\\ & \| c(\cdot,\cdot,\mu) - c(\cdot,\cdot,\tilde {\mu}) \|_w \leq L_1 W_1(\mu,\tilde{\mu})
\\ & \sup_{a, \mu} | c(x,a,\mu)-c(\tilde x, a, \mu) | \leq L_2d_X(x,\tilde x)
\\ & \int_X w_{\max}(y) p(dy|x,a,\mu) \leq \alpha w(x,a)
\\ & \sup_x W_1(p(\cdot|x,a,\mu),p(\cdot|x,\tilde a, \tilde {\mu})) \leq K_1(\|a-\tilde a \|+W_1(\mu,\tilde{\mu}))
\\ & \sup_\mu W_1(p(\cdot|\tilde x,\tilde a,\mu),p(\cdot|\tilde x,\tilde a, {\mu})) \leq K_2(d_X(x,\tilde x)+\|a-\tilde a \|)
\end{align}

\item[b)] For any $(x,Q,\mu) \in X \times \mathcal C_1 \times \mathcal P(X)$, $H^1_1(Q,\mu)(x,\cdot)$ is
 $\rho$-strongly convex.

\item[c)] On $\mathcal C_1,$ for any $(x,Q,\mu),(\tilde x, \tilde Q, \tilde {\mu}) \in X \times \mathcal C_1 \times \mathcal P(X)$, there exists a fixed constant $K_F$ such that the gradient of $H^1_1(Q,\mu)(x,\cdot),$ i.e. $\nabla H^1_1(Q,\mu)(x,\cdot),$ satisfies 
\begin{align*}
& \sup_{a \in A} \| \nabla H^1_1(Q,\mu)(x,a) - \nabla H^1_1(\tilde Q, \tilde{\mu})(\tilde x,a) \| 
\leq K_F(d_X(x,\tilde x)+\|Q_{\min}-\tilde{Q}_{\min} \|_{w_{\max}}+W_1(\mu,\tilde{\mu}))
\end{align*}
\end{itemize}
\end{assumption}

\smallskip

To use the Banach fixed point theorem, we first have to show that the operators $H_{\gamma}$ are well defined. The following lemma makes a connection between fixed points of the operator $H_{\gamma}$ for different $\gamma$ values, which will allow us to show that $H_{\gamma}$ are well defined.

\begin{lemma}\label{lemma:a}
If $\left(Q, \mu (\cdot) =\int_{X}p(\cdot|x,f_{\gamma}(x,Q,\mu),\mu)\mu(dx)\right)\in \mathcal C_{\gamma} \times \mathcal P(X)$ 
is a fixed point of $H_{\gamma},$ then $f_{\gamma}(x,Q,\mu)=f_{\gamma '}(x,\frac {\gamma '}{\gamma} Q,\mu)$ and 
$$\bigg(\frac {\gamma '}{\gamma} Q, \mu(\cdot)=\int_{X \times A}p(\cdot|x,f_{\gamma '}(x, \frac {\gamma '}{\gamma} Q,\mu),\mu)\mu(dx)\bigg) \in \mathcal C_{\gamma '} \times \mathcal P(X)$$ 
is a fixed point of $H_{\gamma '},$ i.e. there exists a fixed point of $H_{\gamma }$ if and only if there exists a fixed point for any other $H_{\gamma '}.$ Furthermore, for any $(Q,\mu) \in \mathcal C_{\gamma} \times \mathcal P(X)$
$$ H_{\gamma}(Q,\mu) = \bigg(\frac {\gamma }{\gamma '} H^{\gamma '}_1\bigg (\frac {\gamma '}{\gamma } Q,\mu\bigg ),H^{\gamma '}_2\bigg (\frac {\gamma '}{\gamma} Q,\mu\bigg )\bigg).$$

\end{lemma}
\begin{proof}
Since 
\[
H^{\gamma}_1(Q,\mu)(x,a) = Q(x,a) = \gamma  c(x, a,\mu) +\beta \int _X Q_{\min}(y)p(dy|x,a,\mu),
\] 
we have
\begin{align}
\frac {\gamma '}{\gamma} H^{\gamma}_1(Q,\mu)(x,a) &
= \gamma ' c(x, a,\mu)  +\beta \int _X \frac {\gamma '}{\gamma} Q_{\min}(y)p(dy|x, a,\mu) \nonumber 
\\& = \frac {\gamma '}{\gamma}  Q(x,a) = H^{\gamma '}_1\bigg (\frac {\gamma '}{\gamma} Q,\mu \bigg)(x,a) \nonumber
\end{align}
and \( 
\frac{\gamma '}{\gamma} Q \in \mathcal C_{\gamma '}.
\)
The function $H^{\gamma}_1(Q,\mu)(x,\cdot)=\gamma H^{1}_1(\frac 1{\gamma}  Q,\mu)(x,\cdot)$ is strictly convex, thereby it has a unique minimizer on $A$. Hence multiplying it with positive constants will preserve the minimizer; that is 
\[
f_{\gamma}(x,Q,\mu) = f_{1}(x,\frac{1}{\gamma}Q,\mu) = f_{\gamma '}(x,\frac {\gamma '}{\gamma} Q, \mu).
\] 
Thus $(\frac {\gamma '}{\gamma} Q , \mu)$ is a fixed point of $H_{\gamma '}.$ 
\end{proof}

As a consequence of the lemma above, we can demonstrate that $H_{\gamma}$ is well-defined and satisfies the conditions in Assumption \ref{assump1} with slightly different Lipschitz coefficients. This will be essential for showing the convergence of iterations of the operators $H_{\gamma}$ later on.

\begin{lemma}\label{welldef}
Suppose Assumption \ref{assump1} holds. Then, for any $\gamma \in \mathbb R_{>0}$, $H_1^{\gamma}$ is well defined and satisfies the conditions in Assumption~\ref{assump1}-(b),(c) with possibly different coefficients $\rho$ and $K_F$.
\end{lemma}
\begin{proof}
Under Assumption \ref{assump1}, the operator $H^1_{1}$ is well defined, i.e. maps $\mathcal C_{1}$ into itself, cf. \cite[Lemma 1]{anahtarci2020value}. This combined with the Lemma \ref{lemma:a} gives us that, for any $\gamma \in \mathbb R_{>0}$, the operator $H^{\gamma}_1$ is well defined. Furthermore, it is easy to see that $H^{\gamma}_1(Q,\mu)(x,\cdot)$ is strongly convex. Moreover, since $H^{\gamma}_1(Q,\mu)=\gamma H^{1}_1(\frac Q{\gamma},\mu)$, we have
\begin{align}
& \| \nabla H^{\gamma}_1(Q,\mu)(x,a)-\nabla H^{\gamma}_1(\tilde Q,\tilde \mu)(\tilde x,a)\|  = \|\nabla \gamma H^1_1(\frac Q{\gamma},\mu)(x,a)-\nabla \gamma H^1_1(\frac {\tilde Q}{\gamma},\mu)(\tilde x,a) \|. \nonumber
\end{align}
Hence, the condition for $\nabla H^{\gamma}_1$ also holds.
\end{proof}

The following lemma is essentially the same as \cite[Theorem 1]{anahtarci2020value} with minor changes. To obtain the Lipschitz coefficients of the components of the scaled operator $H_{\gamma}$, for $\gamma \in \mathbb R_{>0}$, we will give a sketch of the proof adjusting to our setting, details of the calculation can be found in \cite{anahtarci2020value}.

\begin{lemma}\label{lemma:b}
For any $(Q,\mu),(\tilde Q,\tilde{\mu}) \in \mathcal C_{\gamma} \times \mathcal P(X)$, the components of the operator $H_{\gamma}$ satisfy the following Lipschitz bounds:
\begin{align*}
 &\|H^{\gamma}_{1}(Q,\mu)- H^{\gamma}_{1}(\tilde Q,\tilde{\mu})\|_{w} 
  \leq \gamma \bigg(L_1 +\beta\frac{ L_2}{1-\beta K_2} K_1 \bigg)W_1(\mu,\tilde \mu)+\alpha \beta \|Q-\tilde Q\|_{ w}, \\
 &W_1\big(H^{\gamma}_{2}(Q,\mu),H^{\gamma}_{2}(\tilde Q,\tilde{\mu})\big)
  \leq \frac{K_F}{\rho}\frac{K_1}{\gamma} \| Q-\tilde Q\|_{w}+\bigg(\frac{K_F}{\rho}+1\bigg)\bigg(K_1+K_2\bigg)W_1(\mu,\tilde{\mu}).
\end{align*}
\end{lemma}

\begin{proof}
The first inequality can be obtained via triangle inequality as follows:
\begin{align*}
\|H^{\gamma}_{1}(Q,\mu)- H^{\gamma}_{1}(\tilde Q,\tilde{\mu})\|_{w} 
&\leq  \sup_{x,a} \frac{ \gamma | c(x,a,\mu) - c(x,a,\tilde{\mu})|}{ w(x,a) } 
\\& \qquad + \sup_{x,a} \frac{\bigg | \beta\int_XQ_{\min}(y)p(dy|x,a, \mu)-\beta\int_X\tilde Q_{\min}(y)p(dy|x,a,{\mu})\bigg |}{w(x,a)}
\\& \qquad + \sup_{x,a} \frac{\bigg | \beta\int_X\tilde Q_{\min}(y)p(dy|x,a,\mu)-\beta\int_X\tilde Q_{\min}(y)p(dy|x,a, \tilde{\mu})\bigg |}{ w(x,a)}
\\&\leq \gamma L_1 W_1(\mu,\tilde {\mu})
\\& \qquad +\beta \|Q_{\min}-\tilde Q_{\min} \|_{ w_{\max}} \sup_{x,a} \frac{\int_X  w_{\max}(y)p(dy|x,a,\mu)}{w(x,a)}
\\& \qquad +\beta \|\tilde Q_{\min} \|_{\Lip} \sup_{x,a}\frac{W_1(p(\cdot|x,a,\mu),p(\cdot|x,a,\tilde {\mu}))}{w(x,a)}
\\&\leq \gamma L_1 W_1(\mu,\tilde{\mu}) + \beta \alpha \| Q - \tilde Q \|_{ w} +\beta \frac{\gamma L_2}{1-\beta K_2}K_1W_1(\mu,\tilde{\mu}).
\end{align*}

For $Q \in \mathcal C_{\gamma}$, we have $\|\frac 1 {\gamma} Q\|_w \leq \frac{M}{1-\beta \alpha}$ and $\|\frac 1 {\gamma} Q\| \leq \frac{L_2}{1-\beta K_2}$. Thus $\frac 1{\gamma} Q \in \mathcal C_1$, and so, we have the following inequality for $f_{1}(x,\frac 1{\gamma} Q, \mu)$ by \cite[Theorem 1]{anahtarci2020value}
\begin{align*}
 &\|f_{1}(x,\frac 1{\gamma} Q, \mu) -f_{1}(y,\frac 1{\gamma} \tilde Q, \tilde{\mu}) \| 
\leq \frac{K_F}{\rho}(d_X(x,y)+W_1(\mu,\tilde {\mu} ))+\frac 1{\gamma} \frac{K_F}{\rho}\|Q-\tilde Q\|_{w}
\end{align*}
so the same Lipschitz condition also holds for $f_{\gamma}(x,Q, \mu)$ by Lemma \ref{lemma:a}.

To obtain the Lipschitz bound on the operator $H^{\gamma}_2(\cdot,\cdot)=H^1_2(\frac{1}{\gamma} \cdot,\cdot)$, for any $(Q,\mu),(\tilde Q,\tilde {\mu}) \in C_{\gamma} \times \mathcal P(X)$, we have
\begin{align}
W_1(H^1_2 (\frac 1{\gamma} Q,\mu),H^1_2(\frac{1}{\gamma} \tilde Q,\tilde{\mu}))
& = \sup_{\|g\|_{\Lip} \leq 1} \bigg | \int_{X} \int_X g(y)p(dy|x,f_1(x,\frac 1{\gamma} Q,\mu),\mu)\mu(dx)
\nonumber  \\& \qquad \qquad \qquad - \int_{X } \int_X g(y)p(dy|x,f_1(x,\frac 1{\gamma} \tilde Q,\tilde{\mu}),\tilde {\mu})\tilde {\mu}(dx) \bigg |
\nonumber \\& \leq \sup_{\|g\|_{\Lip} \leq 1} \bigg | \int_{X } \int_X g(y)p(dy|x,f_1(x,\frac 1{\gamma} Q,\mu),\mu)\mu(dx)
\nonumber \\& \qquad \qquad \qquad - \int_{X } \int_X g(y)p(dy|x,f_1(x,\frac 1{\gamma} \tilde Q,\tilde{\mu}),\tilde {\mu}) {\mu}(dx) \bigg |
\nonumber \\& \qquad + \sup_{\|g\|_{\Lip} \leq 1} \bigg | \int_{X } \int_X g(y)p(dy|x,f_1(x,\frac 1{\gamma} \tilde Q,\tilde{\mu}),\tilde {\mu}) {\mu}(dx) 
\nonumber \\& \qquad \qquad \qquad - \int_{X } \int_X g(y)p(dy|x,f_1(x,\frac 1{\gamma} \tilde Q,\tilde{\mu}),\tilde {\mu})\tilde {\mu}(dx) \bigg |
\nonumber \\& \leq  \int_{X } \sup_{\|g\|_{\Lip} \leq 1}  \bigg |\int_X  g(y)p(dy|x,f_1(x,\frac 1{\gamma} Q,\mu),\mu)
\nonumber \\& \qquad \qquad \qquad - \int_{X} \int_X g(y)p(dy|x,f_1(x,\frac 1{\gamma} \tilde Q,\tilde{\mu}),\tilde {\mu})\bigg | {\mu}(dx) 
\nonumber \\& \qquad + K_2\bigg(1 +\frac{K_F}{\rho} \bigg)W_1(\mu,\tilde{\mu}) \,\,
\label{eq7} \\&  \leq \int_{X} W_1(p(\cdot|x,f_1(x,\frac 1{\gamma}Q,\mu),\mu),p(\cdot|x,f_1(x,\frac 1{\gamma}\tilde Q,\tilde{\mu}),\tilde \mu)) \mu(dx)
\nonumber \\& \qquad + K_2\bigg(1+\frac{K_F}{\rho} \bigg)W_1(\mu,\tilde{\mu})
\nonumber \\ \leq \frac{K_F}{\rho}K_1&\bigg(\frac 1{\gamma} \|Q-\tilde Q\|_w + W_1(\mu,\tilde{\mu})\bigg) + K_1W_1(\mu,\tilde{\mu}) + K_2\bigg (1+ \frac {K_F}{\rho} \bigg) W_1(\mu,\tilde{\mu}). \nonumber
\end{align}
The only difference between the inequalities above and those in \cite[Theorem 1]{anahtarci2020value} is the term $K_2(K_F/\rho 
+1)$ on \eqref{eq7}. In \cite{anahtarci2020value}, one can easily see that the factor $K_2(K_F/\rho +1)$ can be substituted in place of $K_2+K_F/\rho$ by modifying the proof of \cite[Theorem 1]{anahtarci2020value}. The only difference comes from \cite[Equation 10]{anahtarci2020value} where $K_2(K_F/{\rho}+1)$ can be obtained instead of $K_2+K_F/{\rho}$ which is an improvement since we will be implicitly assuming that $K_2<1.$ This justifies the appearence of the term $K_2(K_F/{\rho}+1).$ This completes the proof.
\end{proof}

Let us define the following constants
\[ 
k_1 = \frac{K_F K_1}{\rho(1-\alpha \beta )},\,\,\,
 k_2 = \frac{1-\bigg ( \frac{K_F}{\rho}+1 \bigg)\bigg (K_2+K_1 \bigg ) }{L_1+\beta\frac{L_2}{1-\beta K_2}K_1}, \,\,\,
k = \max \bigg ( \alpha \beta , \bigg ( \frac{K_F}{\rho}+1 \bigg)\bigg (K_2+K_1 \bigg ) \bigg ). \]

\begin{theorem}\label{thrm1}
 Suppose that $k<1$ and $k_1 < k_2.$ Then there exists a unique fixed point of $H_{\gamma}$ for any $\gamma \in \mathbb R_{>0}$ whenever Assumption \ref{assump1} holds.
\end{theorem}
\begin{proof}
Since $0<k_1<k_2$ there exists $\gamma \in \mathbb R_{>0}$ such that $k_1< \gamma < k_2 $ so we have
$
\frac{K_F K_1}{ \gamma \rho}+\alpha\beta<1
$
and 
$
\gamma(L_1+\beta\frac{L_2}{1-\beta K_2}K_1)+\bigg ( \frac{K_F}{\rho}+1 \bigg)\bigg (K_2+K_1 \bigg ) < 1.
$
As a consequence of Lemma \ref{lemma:b}, we obtain that $H_{\gamma}$ is a contraction, and thus, has a unique fixed point by the Banach fixed point theorem as $H_{\gamma}$ is well defined by Lemma \ref{welldef}.

If we pick any another scaling factor $\gamma ' \in \mathbb R_{>0},$ then Lemma \ref{lemma:a} implies the operator $H_{\gamma '}$ also has a fixed point. If there exists any other fixed point of $H_{\gamma '}$ then this must induce a fixed point of $H_{\gamma}$ as described in Lemma \ref{lemma:a}. Measure component of the induced fixed point for the operator $H_{\gamma}$ is the same as in $H_{\gamma '}$ as stated in Lemma \ref{lemma:a}, so using the uniqueness of the fixed point of $H^{\gamma}_1(\cdot,\mu)$ for fixed $\mu \in \mathcal P(X)$ we see that $H_{\gamma '}$ must have a unique fixed point. 
\end{proof}

\begin{remark}
The primary distinction between our result and \cite[Theorem 2]{anahtarci2020value} is that while \cite[Theorem 2]{anahtarci2020value} assumes $k_1 < 1 < k_2$, we only require the existence of any real number between them. This difference in assumptions is the fundamental reason for introducing scaled versions of the mean-field equilibrium operator and the subsequent analysis. 
\end{remark}

\begin{remark}
The conditions $k_1 < k_2$ and $k < 1$ were previously employed in \cite{JMLR:v24:21-0505} to attain a mean-field equilibrium through an iterative algorithm. This involved determining the optimal $Q$-function for a fixed $\mu$ and iterating the ergodicity operator over these optimal $Q$-functions, which necessitated finding the fixed point of the value iteration algorithm for the iterated state-measure $\mu$ at each step. In our work, we demonstrate that under the same conditions, the same result can be achieved without the need to find the fixed point of the value iteration algorithm at each step of the ergodicity operator. Consequently, there is no requirement to find the optimal $Q$-functions at each iteration, as was done in \cite{JMLR:v24:21-0505}.
\end{remark}

A consequence of the theorem above is that, if there exists a $\gamma \in \mathbb R_{>0}$ such that $H_{\gamma}$ is a contraction operator, then we can prove further that, for any $\gamma' \in \mathbb R_{>0},$ iterations of the operator $H_{\gamma '}$ converges to a fixed point starting from any point in their domain. This is remarkable since the constant $\gamma '$ might not satisfy the assumptions of the Theorem \ref{thrm1} itself. A simple case of such instance is when the constant $\gamma '$ is equal to $2k_2.$ 

We will give the proof of this fact in case when $\gamma ' = 1$ for simplicity. 

\begin{corollary}
Suppose $k_1<k_2$ and $k<1$ holds. Under Assumption \ref{assump1}, iterations of the operator $H_{\gamma}$ converges for any $\gamma \in \mathbb R_{>0}.$
\end{corollary}

\begin{proof}
Since $k_1<k_2,$ there exists $\gamma \in \mathbb R_{>0}$ such that $k_1<\gamma <k_2$ hence by the Banach fixed point theorem iterations of $H_{\gamma}$ converge to the fixed point of $H_{\gamma},$ which in turn induces a fixed point for any $H_{\gamma '},$ $\gamma ' \in \mathbb R_{>0}.$ Recursively define $(Q^n,\mu^n)=H_{\gamma}(Q^{n-1},\mu^{n-1)}).$ Then $H^{\gamma}_1(Q^n,\mu^n)$ converges in $w$-norm to some $Q_{\gamma},$ meaning that 
$\gamma H^1_1(\frac 1{\gamma} Q^n,\mu^n) \rightarrow Q_{\gamma}$ in $w$-norm thus, $H^1_1(\frac 1{\gamma} Q^n,\mu^n) \rightarrow \frac 1{\gamma} Q_{\gamma}$ in $w$-norm. Since $\frac 1{\gamma} Q_{\gamma}$ is the $Q$-function in the fixed point of $H_1$ and $H_2^{\gamma}(Q^n,\mu^n) = H_2^{1}(\frac 1{\gamma} Q^n,\mu^n) \rightarrow \mu$, we must have that the iterations of the operator $H_1$, when acting on $(\frac 1{\gamma} Q_{\gamma}^n,\mu^n)$, is well-defined and convergent to $(\frac 1{\gamma} Q_{\gamma},\mu)$. The same conclusion can be made for other constants $\gamma'$ similarly. Uniqueness of the fixed point of $H_{\gamma}$ concludes the result.
\end{proof}

\begin{corollary}
Under the assumptions of the Theorem \ref{thrm1}, mean-field equilibria given by the operators $H_{\gamma}$, $\gamma \in \mathbb R_{>0}$, are the same, and this common equilibrium is given by the tuple $(\pi, \mu)$, where $\pi(\cdot | x) = \delta_{\underset{a \in A}{\mathrm{argmin}}H^{\gamma}_1(Q,\mu)(x,a)} (\cdot)$ for $(Q,\mu)\in \mathcal C_{\gamma} \times \mathcal P(X)$ such that $H_{\gamma}(Q,\mu) = (Q,\mu).$
\end{corollary}

\begin{proof}
Follows from the definitions.
\end{proof}

\section{Robustness of MFGs} 

Aim of this section will be to show that under some uniqueness conditions, the limit of mean-field equilibria that are obtained under a converging family of transition probabilities exist. The subtle point is the assumption on the uniqueness of the accumulation point that we will use, which is a common assumption in the literature regarding limit theorems involving mean-field equilibrium \cite{cardaliaguet2019master}.

In continuous-time MFGs, the monotonicity condition is a known sufficient condition for the uniqueness of the associated MFE, as discussed in \cite{lasry2007mean}. In discrete time, the counterpart of the monotonicity assumption requires further conditions, as mentioned in \cite[Remark 3]{saldi2022partially}. An alternative condition for achieving the uniqueness of MFE in the discrete-time case is provided in Theorem \ref{thrm1}.

On a product space, it is well known that even if marginals of a sequence of measures converge weakly, these measures might not converge. This is problematic because for a family of mean-field equilibria $(\pi_n,\mu_n),$ if we define $\nu _n(dx,dy)=\pi_n(dy|x)\mu_n(dx),$ then even if $\mu_n$ converges weakly, $\nu _n$ might not converge. Thus we need to control the convergence of both the state-measures $\mu _n$ and conditional distributions $\pi _n.$ For this purpose, and to simplify the proofs, we shall need the following key lemma for our results.

\begin{lemma}\label{keylemma}
Let $A$ be a compact convex subset of $\mathbb R^m$ and $X$ be a Polish space. Consider an $\alpha$-strongly convex continuous family of functions $F^n:X \times A \rightarrow \mathbb R$ in $a$ for each $x$. If $F^n \rightarrow F$ uniformly over compact sets, then $F$ is also $\alpha$-strongly convex  continuous function in $a$ for each $x$ and 
$$\underset{a \in A}{\mathrm{argmin}}F^n(x_n,a) \rightarrow \underset{a \in A}{\mathrm{argmin}}F(x,a)$$ for all $x_n \rightarrow x$ in $X$. 
In general, if a family of functions $\{F_n\}$ converges uniformly over compact sets to a function $F$ with a unique minimizer in $a$ for each $x$, then for any $\nu_n \rightarrow \nu$ weakly in $\P(X)$, we have 
$$
\nu_n \otimes \pi_n \rightarrow \nu \otimes \pi \,\, \text{weakly}
$$
where $\supp \pi_n(\cdot|x) \subset \argmin_{a} F_n(x,a)$ and $\pi(\cdot|x) = \delta_{\argmin_{a} F(x,a)}(\cdot)$, for each $x$.
\end{lemma}

\begin{proof}
Any $\alpha$-strongly convex function $g:A\subset \mathbb R^m \rightarrow \mathbb R$ is of the form $g(a)=\tilde{g}(a)+\alpha /2 \| a \|,$ where $\tilde g$ is a convex function. Fix any $x$. Since $F^n(x,\cdot)$'s are all $\alpha$-strongly convex in $a$, and $F^n(x,\cdot)$ converges in the uniform norm on $A$, $\sup_a|F^n(x,a)-F^m(x,a)| = \sup_a| \tilde F^n(x,a) - \tilde F^m(x,a) | \rightarrow 0$ as $n,m \rightarrow \infty$. Thus $\tilde F^n(x,\cdot) \rightarrow \tilde F(x,\cdot)$ for some convex function $\tilde F(x,\cdot)$ on $A$ in the uniform norm as convex functions are closed under uniform norm. This implies that $\tilde F^n(x,\cdot) + \alpha /2 \|\cdot \|$ converges to $\tilde F(x,\cdot) + \alpha /2 \| \cdot \|$ in uniform norm, and so, $F(x,\cdot) = \tilde F(x,\cdot) + \alpha/2  \| \cdot \|$ is $\alpha$-strongly convex. From strong convexity it follows that $F^n(x,\cdot)$ and $F(x,\cdot)$ have unique minimizers. Please note that if we can establish the second claim, it will also lead to the proof of the first part when we choose $\mu_n = \delta_{x_n}$ and $\mu = \delta_x$. Therefore, for the rest of the proof, we will focus on demonstrating the second assertion.

Note that we have 
\begin{align*}
\nu \otimes \pi \left\{(x,a): F(x,a) = \min_{a} F(x,a) \right\} &= 1 \\
\nu_n \otimes \pi_n \left\{(x,a): F_n(x,a) = \min_{a} F_n(x,a) \right\} &= 1
\end{align*}
due to the definitions of $\pi$ and $\pi_n$. Moreover, since the sequence $\{\nu_n\}$ is tight in $\P(X)$ and $A$ is compact, the sequence $\{\nu_n \otimes \pi_n\}$ is also tight in $\P(X\times A)$. Therefore, there exists a subsequence $\{\nu_{n_k}\otimes \pi_{n_k}\}$ of $\{\nu_n \otimes \pi_n\}$ that converges weakly to $\nu \otimes \hat \pi$. In the remaining part of this proof, we demonstrate that $\nu \otimes \hat \pi = \nu \otimes \pi$. This will conclude the proof, as it implies that every subsequence of ${\nu_n \otimes \pi_n}$ weakly converges to the same limit point. Additionally, due to tightness, we can ascertain the existence of at least one converging subsequence.

To simplify the notation, in the sequel, we denote the converging subsequence as $\{\nu_n\otimes \pi_n\}$. We first prove that 
$$
\nu \otimes \hat\pi \left\{(x,a): F(x,a) = \min_{a} F(x,a) \right\}  = 1 
$$
Indeed, let $C_n \coloneqq \bigl\{ (x,a): F_n(x,a) = \min_{a} F_n(x,a) \bigr\}$. Since both $F_n$ and $\min_{a} F_n(\cdot,a)$ are continuous, $C_n$ is closed. Define $C \coloneqq \bigl\{ (x,a): F(x,a) = \min_a F(x,a) \bigr\}$ which is also closed as both $F$ and $\min_a F(\cdot,a)$ are continuous.

One can prove that $\min_a F_n(\cdot,a)$ converges to $\min_a F(\cdot,a)$ continuously\footnote{The sequence of measurable functions $g_n$ is said to converge to $g$ continuously if $\lim_{n\rightarrow\infty}g_n(e_n)=g(e)$ for any $e_n\rightarrow e$ where $e \in E$. If functions are continuous, then this is equivalent to uniform convergence over compact sets.}, as $n\rightarrow\infty$. For each $M\geq1$, define $B^M \coloneqq \bigl\{ (x,a): F(x,a) \geq \min_a F(x,a) + \epsilon(M) \bigr\}$ which is closed, where $\epsilon(M) \rightarrow 0$ as $M\rightarrow \infty$. Since both $F$ and $\min_a F(\cdot,a)$ is continuous, we can choose $\{\epsilon(M)\}_{M\geq1}$ so that $\nu \otimes \hat \pi(\partial B^M) = 0$ for each $M$. Since $C^c = \bigcup_{M=1}^{\infty} B^M$ and $B^M \subset B^{M+1}$, we have by monotone convergence theorem
\begin{align}
\nu_n \otimes \pi_n\big(C^c \cap C_n\big) = \liminf_{M\to\infty} \nu_n \otimes \pi_n \big(B^M \cap C_n \big) \nonumber
\end{align}
Hence, we have
\begin{align}
1 &= \limsup_{n\rightarrow\infty} \liminf_{M\rightarrow\infty} \biggl\{\nu_n \otimes \pi_n\big(C \cap C_n\big) + \nu_n \otimes \pi_n\big(B^M \cap C_n\big)\biggr\} \nonumber\\
&\leq \liminf_{M\rightarrow\infty} \limsup_{n\rightarrow\infty}  \biggl\{\nu_n \otimes \pi_n\big(C \cap C_n\big) + \nu_n \otimes \pi_n\big(B^M \cap C_n\big)\biggr\}. \nonumber
\end{align}
\noindent For fixed $M$, let us evaluate the limit of the second term in the last expression as $n\rightarrow\infty$. First, note that $\nu_n \otimes \pi_n$ converges weakly to $\nu \otimes \hat \pi$ as $n\rightarrow\infty$ when both measures are restricted to $B^M$, as $B^M$ is closed and $\nu \otimes \hat \pi(\partial B^M)=0$ \cite[Theorem 8.2.3]{Bog07}. Furthermore, $1_{C_n \cap B^M}$ converges continuously to $0$: if $(x^{(n)},a^{(n)}) \rightarrow (x,a)$ in $B^M$, then
\begin{align}
\lim_{n\rightarrow\infty} F_n(x^{(n)},a^{(n)}) &= F(x,a)  \,\, \text{(by uniform convergence over compact sets)}\nonumber \\
&\geq \min_a F(x,a)+ \epsilon(M) \nonumber \\
&= \lim_{n\rightarrow\infty} \min_a F_n(x^{(n)},a) + \epsilon(M). \nonumber
\end{align}
Hence, for large enough $n$'s, we have $F_n(x^{(n)},a^{(n)}) > \min_a F(x^{(n)},a)$ which implies that $(x^{(n)},a^{(n)}) \not\in C_n$. Then, by \cite[Theorem 3.5]{Lan81}, for each $M$ we have $\limsup_{n\rightarrow\infty} \nu_n \otimes \pi_n\big(B^M \cap C_n\big) = 0$.
Therefore, we obtain
\begin{align*}
1 \leq  \limsup_{n\rightarrow\infty} \nu_n \otimes \pi_n\big(C \cap C_n\big) \le \limsup_{n\rightarrow\infty} \nu_n \otimes \pi_n(C) \leq \nu \otimes \hat \pi(C), \nonumber
\end{align*}
where the last inequality follows from Portmanteau theorem \cite[Theorem 2.1]{Bil99} and the fact that $C$ is closed. Hence, $\nu \otimes \hat \pi(C)=1$.

Now, let us complete the proof in view of the above result. Note that 
$$
\nu \otimes \hat\pi \left\{(x,a): F(x,a) = \min_{a} F(x,a) \right\}  = 1 
$$
Hence, for $\nu$-a.e., we have 
$$
\hat \pi\left(\left\{(x,a): F(x,a) = \min F(x,a)\right\}\bigg|x\right) = 1
$$
But note that $\argmin_a F(x,a)$ is unique for all $x$, and so, for the $x$-values, where above equality is true, we have
$$
\supp \hat \pi(\cdot|x) = \argmin_a F(x,a)
$$
This implies that $\hat \pi(\cdot|x) = \pi(\cdot|x)$ $\nu$-a.e. Hence, $\nu \otimes \hat \pi = \nu \otimes \pi$. This completes the proof.
\end{proof}

\subsection{Robustness of Value Iteration Algorithm} 

In this section, we show that if we have MFGs with converging transition probabilities under the assumption that value iteration algorithm works for the space of Borel measurable $Q$ functions with $w_{max}$-bounded minima and the nominal model has an unique equilibrium, then mean-field equilibria obtained for the MFGs converge to that of the limiting case.

We will define the space of Borel measurable funtions on \( X \times A \) such that their infimum on $A$ are bounded by $w_{\max}$ as follows:
\[
B_{b,w_{\max}}(X\times A)= \{f \in B(X \times A): \|f_{\min}\|_{w_{\max}} \leq C  \text{ for some constant \(C\) }\}.
\]
We will equip $B_{b,w_{\max}}(X \times A)$ with the $w$-norm.
 
To simplify the notation, we will assume that our mean-field games differ only in their transition probabilities, which we will denote as $p_n$. A similar analysis can be carried out for the case of varying cost functions $c_n$ under the condition that $c_n$ continuously converges to $c$. 

We will denote each game by MFG$_n$ and corresponding mean-field equilibrium operator as 
\begin{align*}
H_{1,p_n}&: B_{b,w_{\max}}(X \times A) \times \mathcal P(X) \rightarrow B_{b,w_{\max}}(X \times A) \times \mathcal P(X) 
\\ & :(Q,\mu) \rightarrow (H^{1}_{1,p_n}(Q,\mu),H^{1}_{2,p_n}(Q,\mu)),
\end{align*}
 where the iterations of the $Q$-functions is handled by the operator
$$H^{1}_{1,p_n}(Q,\mu)(x,a) = c(x,a,\mu) +\beta \int _X Q_{\min}(y)p_n(dy|x,a,\mu)$$ and assuming that a measurable minimizer 
\begin{equation}\label{minchoice} f_{1,p_n}(x,Q,\mu) \in \underset{a \in A}{\mathrm{argmin}} H^{1}_{1,p_n}(Q,\mu)(x,a),
\end{equation}
exists (which is the case under our assumptions), we define the ergodicity operator to handle the iterations on the state-measures
$$H^{1}_{2,p_n}(Q,\mu)(\cdot) = \int_{X} p_n(\cdot|x, f_{1,p_n}(x,Q,\mu),\mu)\mu(dx).$$
We remark that as in Lemma \ref{welldef}, any $Q$-function in $B_{b,w_{\max}}(X \times A)$ is mapped to a $Q$-function in $B_{b,w_{\max}}(X \times A)$ by the operator $H^1_{1,p_n}.$
\begin{theorem}\label{thrm}
Suppose the following conditions hold:
\begin{itemize}
\item[i)] The state space $X$ is locally compact.
\item[ii)] $(p_n)_n$ converges continuously to a transition probability $p$ in the $W_1$-topology; that is, $p_n(\cdot|x^n,a^n,\mu^n) \rightarrow p(\cdot|x,a,\mu)$ with respect to $W_1$ whenever $(x^n,a^n,\mu^n) \rightarrow (x,a,\mu)$ in $X \times A \times \mathcal P(X)$.
\item[iii)] We also have 
\[
\int_X w_{max}(y)p_n(dy|x^n,a^n,\mu^n) \rightarrow \int_X w_{max}(y)p(dy|x,a,\mu)
\]
whenever $(x^n,a^n,\mu^n) \rightarrow (x,a,\mu)$ in $X \times A \times \mathcal P(X).$
\item[iv)] Assumption \ref{assump1} holds for the nominal model.
\item[v)] Iterations of MFE operators $H_{1,p_n}$ for MFG$_n$'s converge to fixed points on $B_{b,w_{max}}(X \times A)\times \mathcal P(X),$ where $\mathcal P(X)$ is equipped with the $W_1$-metric and $B_{b,w_{max}}(X \times A)$ is equipped with the $w$-norm. The same also holds for the nominal model, however we assume that the fixed point is unique for the nominal model.
\end{itemize}
Then, the fixed point of the value iteration algorithm with starting from a common initial datum $(Q_0,\mu_0)$ obtained under transition probability $p_n,$ say $(Q^n_{*},\mu^n_{*}),$ converges to $(Q_*,\mu_*),$ the fixed point obtained under the transition probability $p$ with the same initial datum, in the sense that $Q^n_*(x^n,\cdot)$ converges uniformly to $Q_*(x,\cdot)$ whenever $x_n \rightarrow x$ in $X$, and $\mu^n_*$ converges to $\mu_*$ weakly. In particular, mean-field equilibria $(\mu_*^n,\pi_n^*)$ obtained through the value iteration algorithm under the transition probabilities $p_n$ converges to the one $(\mu_*,\pi_*)$ obtained under the transition probability $p$ under the weak topology on $\mathcal P(X \times A)$; that is, $\mu_*^n \otimes \pi_*^n \rightarrow \mu_*\otimes\pi_*$ weakly.
\end{theorem}

\begin{proof}
\proofpart{1}{Uniform convergence of approximated family over actions:} 
For each $n,$ we define
\[
F(x,Q,\mu,a,p_n)=c(x,a,\mu)+\beta \int_X Q_{\min}(y)p_n(dy|x,a,\mu).
\]
We also define
\[
F(x,Q,\mu,a,p)=c(x,a,\mu)+\beta \int_X Q_{\min}(y)p(dy|x,a,\mu)
\]
for the nominal model. Assume that we have $(Q^n,\mu ^n) \rightarrow (Q,\mu)$; that is, $\mu^n$ converges to $\mu$ weakly in $\mathcal P(X)$ and $Q^n_{\min}$ converges to $Q_{\min}$ continuously.
We further suppose that $\|Q^n_{\min}\|_{w_{\max}}\leq C_n$ and $\|Q_{\min}\|_{w_{\max}}\leq C_{\infty}$. Our aim is to show that 
\[ F(\cdot,Q^n,\mu^n,\cdot,p_n) \rightarrow F(\cdot,Q,\mu,\cdot,p) \] continuously which implies that  
\[
\min_a F(x^n,Q^n,\mu^n,a,p_n) \rightarrow \min_a F(x,Q,\mu,a,p),
\]
whenever $x_n \rightarrow x$ in $X$. Let $K \subset X$ be compact, then triangle inequality gives us the following
\begin{align*}
& \sup_{x,a \in K \times A} |F(x,Q^n,\mu^n,a,p_n) - F(x,Q,\mu,a,p)|
\\&=\sup_{x,a \in K \times A}  \bigg |c(x,a,\mu^n) + \beta \int_X Q^n _{\min}(y)p_ n(dy|x ,a ^n,\mu ^n) -  c(x,a,\mu) 
\\& \qquad \qquad \quad- \beta \int_X Q_{\min}(y)p(dy|x,a,\mu) \bigg|
\\&\leq \sup_{x,a \in K \times A} \bigg | c(x,a,\mu ^n)-c(x,a,\mu) \bigg|
\\&\quad +\beta \sup_{x,a \in K \times A} \bigg | \int_XQ^n_{\min}(y)p _n(dy|x,a,\mu^n)- \int_X Q_{\min}(y)p(dy|x,a,\mu) \bigg |.
\end{align*}
To see the vanishing of the first term on the last line, observe that the cost function $c(\cdot,\cdot,\mu^n)$ converges to $c(\cdot,\cdot,\mu)$ continuously as $n \rightarrow \infty$ by Assumption \ref{assump1}-(a). Therefore, it is sufficient to analyse the vanishing of the second term on the last line.
We will use the following triangle inequality on the second term as follows

\begin{align*}
&\sup_{x,a \in K \times A} \bigg | \int_XQ^n_{\min}(y)p _n(dy|x,a,\mu^n)- \int_X Q_{\min}(y)p(dy|x,a,\mu) \bigg |
\\&=\sup_{x,a \in K \times A} \bigg | \int_X Q^n_{\min}(y)-Q_{\min}(y)p _n(dy|x,a,\mu^n) 
\\&\qquad \qquad \quad + \int_X Q_{\min}(y)(p _n(dy|x,a,\mu^n)-p(dy|x,a,\mu)) \bigg | 
\\&\leq \sup_{x,a \in K \times A} \bigg | \int_X (Q^n_{\min}(y)-Q_{\min}(y))p _n(dy|x,a,\mu^n) \bigg | 
\\&\quad+\sup_{x,a \in K \times A} \bigg | \int_X Q_{\min}(y)(p_n(dy|x,a,\mu^n)-p(dy|x,a,\mu))\bigg|.
\end{align*}
By assumption $Q^n_{\min}$ converges to $Q_{\min}$ continuously, $\|Q^n_{\min}\|_{w_{\max}}\leq C_n$, $\|Q_{\min}\|_{w_{\max}}\leq C_{\infty}$, and 
\begin{align*}
&\int_X w_{\max}(y)p_n(dy|x^n,a^n,\mu^n) \rightarrow \int_X w_{\max}(y)p(dy|x,a,\mu)
\end{align*}
as $n \rightarrow \infty$ whenever $(x^n,a^n)\rightarrow (x,a)$ in $X \times A$.  Since $Q^{n}_{\min}$ converges to $Q_{\min}$, we have $\sup_{n \in \mathbb N \cup \{\infty\}} C_n < \infty$, and so, as a consequence of \cite[Theorem 3.3]{MR705462},  we have 
\begin{align*}
\lim_n \int_X Q^n_{\min}(y)p_n(dy|\cdot,\cdot,\mu^n) = \int_X Q_{\min}(y)p(dy|\cdot,\cdot,\mu)
= \lim_n \int_X Q_{\min}(y)p_n(dy|\cdot,\cdot,\mu^n)
\end{align*}
where the last limit follows from an identical reasoning and the convergences are continuously in $(x,a)$ (and so, uniformly over compact sets). Thus the terms in the inequality above vanish. Therefore, $F(\cdot,Q^n,\mu^n,\cdot,p_n)$ converges to continuously in $(x,a)$ to $F(\cdot,Q,\mu^n,\cdot,p)$.

We conclude this step with the remark that uniform convergence over $a$ implies that $\min_a F(\cdot,Q^n,\mu^n,a,p_n)$ converges continuously to $\min_a F(\cdot,Q,\mu,a,p)$. Indeed, we have the inequality

\begin{align*}
\sup_{x \in K} |\min_aF(x,Q^n,\mu^n,a,p_n)- \min_a F(x,Q,\mu,a,p)|
\leq \sup_{x \in K}\sup_a |F(x,Q^n,\mu^n,a,p_n)- F(x,Q,\mu,a,p)|,
\end{align*}
and so, from the argument above, continuous convergence follows. 
Using \( \int_X Q^n_{\min}(y)p_n(dy|x,a,\mu)\leq C_n w_{max}(x), \) and Assumption~\ref{assump1}-(a), we obtain
\[
\min_a F(x^n,Q^n,\mu^n,a,p_n) \leq \tilde C_n w_{\max}(x^n).
\]
A similar type of bound can be found on $F(x,Q,\mu,a,p)$ for any $Q$ with $\|Q_{\min}\|_{w_{\max}}\leq C_{\infty}$, which means that $w_{\max}$ will bound our iterations. Hence, we can repeat this sort of convergence over $\tilde Q^n(x,a) = F(x,Q^n,\mu^n,a,p_n)$ and $\tilde Q(x,a) = F(x,Q,\mu,a,p)$ on the next iteration.

\proofpart{2}{Convergence of iterations to each other:}
For the fixed common initial datum $(Q_0,\mu_0),$ we have the same $Q_{0,\min}$ for each $n$ by definition, and  so, it is sufficient to prove that $\|Q_{0,\min}\|_{w_{\max}} \leq C $ by some constant $C.$ Since $1 \leq w,$ from the inequality
\begin{align*}
 |Q_{0,\min}(x)| \leq w_{\max}(x) \frac{ |Q_{0,\min}(x)| } { w_{\max}(x) } \leq w_{\max}(x) \| Q_{0,\min} \|_{w_{\max}} < \infty
\end{align*} we obtain the bound $C = \|Q_{0,\min} \|_{w_{\max}}.$
This means Step 1 is effective over our iterations starting from a common initial datum.

From the Step 1, we have that $Q^n_1(x^n,a) = F(x^n,Q_0,\mu_0,a,p^n)$ converges to $Q_1(x,a) = F(x,Q_0,\mu_0,a,p)$ uniformly over $a$ and $Q^n_{1,\min}(x^n)\rightarrow Q_{1,\min}(x)$ whenever $x^n\rightarrow x$ in $X.$ Note that by Assumption \ref{assump1}, minimum of $Q_1(x,a)$ is unique in $a$. Then by Lemma ~\ref{keylemma}, any $f^n(x^n,Q_0,\mu_0) \in \underset{a \in A}{\mathrm{argmin}} Q^n_1(x^n,a)$ converges to $f(x,Q_0,\mu_0)=\underset{a \in A}{\mathrm{argmin}} Q_1(x,a)$ as $n \rightarrow \infty$ in $\mathbb R^m$  for any $x^n \rightarrow x$ in $X.$

Recall that iteration of $\mu_0$ with respect to the transition probability $p_n$ in value iteration algorithm is of the form
$\mu^n_1(\cdot)= \int_{X} p_n(\cdot|x,f^n(x,Q_0,\mu_0),\mu_0)\mu_0(dx).$ We shall show that $\mu^n_1$ converges weakly to 
$$\mu_1=\int_{X} p(\cdot|x,f(x,Q_0,\mu_0),\mu_0)\mu_0(dx).$$
To this end, pick arbitrary real valued bounded continuous function $g$ over $X.$ Then, from Jensen's inequality, we get

\begin{align*}
& \bigg | \int_{X} \int_{X} g(y)p_n(dy|x,f^n(x,Q_0,\mu_0),\mu_0)\mu_0(dx) 
- \int_{X } \int_{X} g(y)p(dy|x,f(x,Q_0,\mu_0),\mu_0)\mu_0(dx) \bigg | \\
& \leq  \int_{X} \bigg | \int_X g(y) [p_n(dy | x, f^n(x,Q_0,\mu_0),\mu_0)-p(dy | x, f(x,Q_0,\mu_0)] \bigg | \mu_0(dx).
\end{align*}
Since $p_n$ converges to $p$ weakly and continuously, we obtain that $\mu ^n_1$ converges to $\mu_1$ weakly \cite[Theorem 3.3]{MR705462}. One can now proceed via induction by repeating the same argument mutatis mutandis since the convergence results we use will be sufficient. Suppose that up until $k$-th iteration, $(Q^n_k,\mu^n_k)$ converges to $(Q_k,\mu_k)$ as in the statement. From Step 1 and Lemma \ref{keylemma}, we have $f^n(x^n,Q^n_k,\mu^n_k) \rightarrow f(x,Q_k,\mu_k)$ whenever $x^n$ converges to $x$ in $X.$ Sufficient to show that $\mu^n_{k+1} \rightarrow \mu_{k+1}$ weakly as Step 1 takes care of convergence of the iteration of $Q^n_{k+1}$ and $Q_{k+1}$. For an arbitrary real valued bounded continuous function $g$, we have
\begin{align*}
&\bigg | \int_{X} \int_X g(y) p_n(dy|x,f^n(x,Q^n_k,\mu^n_k),\mu^n_k)\mu^n_k(dx) 
 - \int_{X} \int_X g(y)p(dy|x,f(x,Q_k,\mu_k),\mu_k)\mu_k(dx) \bigg |
\\& \leq \bigg | \int_{X} \int _X g(y)p_n(dy|x,f^n(x,Q^n_k,\mu^n_k),\mu^n_k)[\mu^n_k(dx)-\mu_k(dx)]\bigg|
\\&\quad +\int_{X}\bigg | \int _X g(y) \bigg[p_n(dy|x,f^n(x,Q^n_k,\mu^n_k),\mu^n_k)-p(dy|x,f(x,Q_k,\mu_k),\mu_k)\bigg]\bigg | \mu_k(dx).
\end{align*}
Note that as in the base case of the  induction, the following convergence holds
\[
\int_X g(y)p_n(dy|x^n,f^n(x^n,Q^n_k,\mu^n_k),\mu^n_k) \rightarrow \int_X g(y)p(dy|x,f(x,Q_k,\mu_k),\mu_k).
\]
This implies that
\begin{align*}
\lim_n \int_{X} \int_X g(y)p_n(dy|x,f^n(x,Q^n_k,\mu^n_k),\mu^n_k)\mu^n_k(dx) 
&= \lim_n \int_{X} \int_X g(y)p_n(dy|x,f^n(x,Q^n_k,\mu^n_k),\mu^n_k)\mu_k(dx)
\\&= \int_{X} \int_X g(y)p(dy|x,f(x,Q_k,\mu_k),\mu_k)\mu_k(dx)
\end{align*}
once again by \cite[Theorem 3.3]{MR705462}. Hence, we can conclude that the terms in the above inequality vanish. Thus, $\mu^n_{k+1}$ converges to $\mu_{k+1}$ weakly and we are done with the induction.

\proofpart{3}{Convergence of the mean-field equilibria to the limiting case}

Based on the assumptions we have made, the value iteration algorithm yields a mean-field equilibrium under any transition probability $p_n$ and $p$ as a consequence of assumption (v) in the theorem statement.

Let 
\(
(Q_{*},\mu_{*})
\) 
be the such MFE obtained under 
\(
p
\)
and 
\(
(Q ^n_*, \mu ^n _*)
\) 
be the MFE obtained under 
\(
p_n
\) 
through value iteration algorithm. Similarly, let 
\(
(Q_{k},\mu_{k})
\) 
and 
\(
(Q^n_{k},\mu^n_{k})
\)
be the \(k \)-th iteration of the value iteration algorithm starting from \((Q_0,\mu_0)\) under the transition probability $p$ and $p_n$, respectively. We shall show that $Q^n_*$ converges $Q_*$ continuously over $X \times A.$ Indeed, we have
\begin{align*}
 \sup_a |Q^n_*(x^n,a)-Q_*(x,a)| &\leq \sup_a |Q^n_*(x^n,a)-F(x^n,Q^n_k,\mu^n_k,a,p_n)|
 \\&\qquad + \sup_a | F(x^n,Q^n_k,\mu^n_k,a,p_n)-F(x,Q_k,\mu_k,a,p)|  
 \\&\qquad + \sup_a | F(x,Q_k,\mu_k,a,p) - Q_*(x,a)|.
\end{align*}
Since $(Q^n_*)_n$ and $Q_*$ are fixed points obtained under $w$-norm and $w$ is a continuous strictly positive function bounded away from $0$, for any family of continuous functions \(f^n:X \times A \rightarrow \mathbb R \) and fixed \(x \in X \)
\[ 
\frac{\sup_a |f^n(x,a)|}{\sup_a w(x,a)} \leq \sup_a\frac{|f^n(x,a)|}{w(x,a)} \leq \|f^n \|_{w}
\]
so if \(f^n \)'s converge to \(0\) in \(w \) norm, they also converge to \(0\) continuously. Combined with convergence of $\mu^n_k$ to $\mu^n_*$ and $\mu_k$ to $\mu_*$ in $W_1$ metric, which is stronger than that of weak convergence, implies that the first and last terms vanish as $n \rightarrow \infty$.
  As a direct consequence of Step 2, the second term also vanishes. Thus, from Lemma \ref{keylemma} as the minimum of $Q_*(x,a)$ is unique in $a$, we obtain that $f^n(x^n,Q^n_*,\mu^n_*) \rightarrow f(x,Q_*,\mu_*)$ whenever $x^n \rightarrow x$ in $X$ and consequently 
$
\pi^n_*(\cdot|x^n)=\delta_{f^n(x^n,Q^n_*,\mu^n_*)}(\cdot)
$
converges weakly (and so continuously) to 
\(
\pi_*(\cdot|x) = \delta_{f(x,Q_*,\mu_*)}(\cdot).
\)
Repeating the proof of Step 2, we see that $\mu ^n_*$ converges weakly to $\mu _*$ meaning that the mean field equilibrium $(\pi^n_*,\mu^n_*)$ obtained under $p^n$ through the value iteration algorithm converges to the one obtained under $p$, $(\pi _*,\mu _*).$ 
\end{proof}


\begin{remark}
For the robustness of the value iteration algorithm, it suffices to have Assumption \ref{assump1} for the nominal model in addition to assumption (v) in Theorem~\ref{thrm}. However, to ensure the convergence of the value iteration algorithm for the MFG$_n$'s, they generally need to satisfy Assumption \ref{assump1} and the condition in Theorem~\ref{thrm1}. This is why we established the convergence of the value iteration algorithm in detail in the previous section.

In the next section, we will address the finite model approximation problem. While the quantized model with a finite state space satisfies Assumption \ref{assump1} and the condition in Theorem~\ref{thrm1} (resulting in value iteration convergence), the extended version of the quantized model, which shares the same state space as the nominal model, does not satisfy these assumptions.

However, in this case, the extended game can inherit the conditions required for assumption (v) in Theorem~\ref{thrm} from the quantized model. 
\end{remark}

By applying Lemma \ref{keylemma} to the iterations on the approximating MFGs, we can deduce that the Dirac measure induced by the minimizer of the iterations also converges to the optimal policy in the equilibrium of the nominal model. This result is a consequence of the proof of the theorem mentioned earlier. 
 
\begin{theorem}\label{thrm:a}
Suppose the nominal model and each MFG$_n$ satisfies the conditions in Theorem \ref{thrm}. Further assume that $Q_*^n(x,a)$ has a unique minimizer in $a$ for each $x$, where $(Q_n^*,\mu_*^n)$ is the fixed point of $H_{1,p_n}$ in the previous theorem. Let $(\pi_*,\mu_*)$ be the mean-field equilibrium of the nominal model. If $(Q^{m,n},\mu^{m,n})$ is the $m$'th iteration of the value iteration algorithm for MFG$_n$, then for $\pi^{m,n}$ satisfying $ \supp \pi^{m,n}(\cdot|x) \subset \argmin_a Q^{m,n}(x,a)$, we have 
\[
W_1(\pi^{m,n}(\cdot|x_n),\pi _*(\cdot|x))< \epsilon
\]
and
\[
W_1(\mu^{m,n},\mu _*)< \epsilon
\]
for sufficiently large $m$ and $n$, where $(x_n)_n \subset X$ such that $x_n \rightarrow x$. In other words, $(\pi^{m,n},\mu^{m,n})$ is an $\epsilon$-near mean-field equilibrium for the nominal model.
\end{theorem}
\begin{proof}
Let $(\mu^{n}_{*},\pi^{n}_{*})$ be a mean-field equilibrium obtained through value iteration algorithm for MFG$_n$ with the fixed point $(Q^n_*,\mu^n_*).$ It follows from Theorem \ref{thrm} that
\[ \lim_n W_1(\pi^n_*(\cdot|x_n),\pi_*(\cdot|x)) = 0 \]
and
\[ \lim_n W_1(\mu^n_*,\mu_*) = 0.\]
As we have shown in Theorem \ref{thrm}, since the iterations of the $Q$-functions converge to the fixed point in $w$-norm, under the uniform norm, we have $\lim_{m} \| Q^{m,n}(x,a)-Q^n(x,a) \|_{\infty} = 0$, and so, as a consequence of Lemma \ref{keylemma}, we have
\[ \lim_m W_1(\pi^{m,n}(\cdot|x_n),\pi^n_*(\cdot|x)) =0. \]
Similarly, as a consequence of the value iteration algorithm, $\mu^{m,n}$ converges to $\mu^n_*$ under the $W_1$ topology, and so, we also have
\[ \lim_m W_1(\mu^{m,n},\mu^n_*) = 0. \]
The result then follows from the triangle inequality.
\end{proof}

In particular, when each MFG$_n$ is equivalent to the nominal model, the theorem above shows that value iteration algorithm allows us to find $\epsilon$-near mean-field equilibrium for the nominal model. 

\section{An Application: Quantization of Mean-field Games in Compact State Spaces}

In this section, we apply the robustness findings established in the previous section to the quantization of mean-field games. It is well-known that when the state and action spaces are uncountable, the need arises to discretize these spaces. This discretization is essential for approximating the mean-field equilibrium, as directly computing the mean-field equilibrium using the mean-field equilibrium operator $H_1$ becomes infeasible in such cases. In this section, we treat the quantized model as a perturbed variant and utilize our findings to demonstrate that with a sufficiently fine quantization, the mean-field equilibrium of the quantized model can be used to closely approximates the mean-field equilibrium of the nominal model.

\subsection{Construction of Quantized Mean-field Games} \label{q-mfgs}

In this subsection, we assume that our state space $X$ is compact. Since $X$ is compact, we can find a sequence $(\{x_{n,i}\}_{i=1}^{k_n})_{n\geq 1}$ of finite grids in $X$ such that for all $n$ and $x \in X,$
\[\min_{i \in \{1,2, \cdots,k_n\}} d_X(x,x_{n,i})<\frac 1n .\]
Let $X_n=\{x_{n,i}\}_{i=1}^{k_n} \subset X$ and define $T_n:X \rightarrow X_n$ via the relation
\[T_n(x)=\underset{x_{n,i}\in X}{\mathrm{argmin}} d_X(x,x_{n,i}),
\]
where ties are broken so that $T_n$ is measurable. The map $T_n$ is called a quantization map (or a state aggregation map) and each element of $X_n$ is called a representative (for the quantization). We use pullbacks of $T_n$ to partition $X$ as $\{S_{n,i}\}_{i=1}^{k_n}$, where each $S_{n,i}$ is given by
$$ S_{n,i} = \{x \in X : T_n(x)=x_{n,i} \} .$$
Since $T_n$ is measurable, each $S_{n,i}$ is a measurable subset of $X$ and has the diameter $\sup_{x,y \in S_{n,i}} d_X(x,y) < \frac 2n .$ 

Since $X_n=\{x_{n,i}\}_{i=1}^{k_n} \subset X$ is a discrete space, we will equip it with the discrete metric $d_{X_n}:$
\[ 
d_{X_n}(x_{n,i},x_{n,j})=1_{i \not = j}.
\]
Since $X_n$ is finite, the metric space $(X_n,d_{X_n})$ is also a Polish space. Furthermore as $d_{X_n}$ is bounded above by $1$, to be able to obtain similar type of Lipschitz coefficients for the quantized game, we will scale the original metric $d_X.$ For the original metric $d_X$ on $X,$ by scaling it with a constant $1/M,$ where $M=\max_{x,y \in X} d(x,y),$ i.e. considering $1/M d_X$ on $X,$ if one repeats everything as in Lemma \ref{lemma:b}, it can be seen that the conditions for the Theorem \ref{thrm1} is unchanged for the original problem. This means that such scaling on the metric of $X$ will not affect the results obtained for the nominal model due to the existing Lipschitz coefficients.  In case we do not assume that $d_X \leq 1,$ an alternative metric that can be used on the space $X_n$ is $d_{X_n}(x_{n,i},x_{n,j})=M1_{i \not = j},$ where $M$ is the bound on $d_X.$ However, since the scaling on the metric $d_X$ does not require any further assumption regarding the tools we want to utilize, atleast on the nominal model, to simplify the presentation we will work under the convention that $d_X \leq 1.$ Thereby, without loss of generality, by abusing the notation we will assume that the metric $d_X$ is bounded above by $1$, and so, the metric $d_{X_n}$ dominates the metric $d_X$ when $d_X$ is restricted to $X_n \times X_n.$

To obtain the associated Lipschitz coefficients for the quantized game, we require a dominating metric on the quantized metric space (see Lemma \ref{lemma:c} below). The main challenge arises from the fact that when considering any metric on the discretized space that is a multiple of the one in the continuous state space, the controllability of Lipschitz coefficients for variations in the state space is no longer guaranteed. In fact, in such cases, it is possible to choose a metric where the Lipschitz coefficient can become unmanageable (see Remark \ref{rem:a} below), making it impossible to satisfy the assumptions in Theorem \ref{thrm1} for the quantized game, even if the nominal model satisfies them.

Let $\{ \nu _n\}_n$ be a family of probability measures on $X$ with the property that
$\nu_n(S_{n,i}) > 0$
for each $i$ and $n.$ We create further probability measures out of this family as follows: for each $S_{n,i}$ we restrict $\nu_{n}$ to the Borel $\sigma$-algebra restricted on $S_{n,i}$ and define
$\nu_{n,i}(\cdot) = \frac{\nu_n(\cdot)}{\nu_n(S_{n,i})} .$
For each $n,$ using the family $(\nu_{n,i})_{i=1}^{k_n}$, we define  one-stage cost function and transition probability on $X_n \times A \times \mathcal P(X_n)$ by
\begin{align*}
&c_n(x_{n,i},a,\mu^{\mathcal D}) = \int_{S_{n,i}} c(z,a,\mu^{\mathcal C}) \nu_{n,i}(dz) 
\\ & p_n(\cdot|x_{n,i},a,\mu^{\mathcal D}) = \int_{S_{n,i}} T_n \star p(\cdot|z,a,\mu^{\mathcal C}) \nu_{n,i}(dz),
\end{align*}
where $\mu^{\mathcal C}:\mathcal B(X) \rightarrow [0,1]$ is an atomic probability measure on $X$ that is completely characterized by $\mu^{\mathcal D}$ via the relation
\begin{equation}\label{disctocontprob}
 \mu^{\mathcal C}(\bullet) = \sum_{i=1}^{k_n} \mu^{\mathcal D}(x_{n,i})\delta_{x_{n,i}}(\bullet), 
\end{equation}
and $T_n\star p(\cdot|z,a,\mu^{\mathcal C})$ is the pushforward of the measure $p$ with respect to $T_n.$ For a given probability measure $\mu^{\mathcal C}:\mathcal B(X) \rightarrow [0,1],$ similarly we can define a discrete probability measure on $X_n,$ $\mu^{\mathcal D}:\mathcal B(X_n)\rightarrow [0,1],$ by 
\begin{equation}\label{conttodiscprob}
\mu ^{\mathcal D}(\bullet )= \sum_{i=1}^{k_n} \mu^{\mathcal C}(S_{n,i})\delta _{{x_{n,i}}}(\bullet).
\end{equation}
Prominence of these relations on continuous space and discrete space probability measures is that they will allow us to construct discrete approximations of the initial mean-field game as they will allow us to switch between discrete and continuous state space mean-field games as we will see later on.

Using these structures, we define a quantized mean-field game on the state space $X_n$ and action space $A$  via the transition probability $p_n$ and the one-stage cost function $c_n$. We also define $H_{1,n}$ as follows
\begin{align*}
H_{1,n}&:\mathcal C(n) \times \mathcal P(X_n) \rightarrow \mathcal C(n) ,
\\&:(Q,\mu^{\mathcal D})\rightarrow c_n(x,a,\mu^{\mathcal D})+\beta \sum_{y \in X_n} Q_{\min}(y)p_n(y|x,a,\mu^{\mathcal D}),
\end{align*}
where 
\[\mathcal C(n) = \bigg \{Q:X_n \times A \rightarrow [0,\infty) : \| Q \|_{\infty} \leq \frac M{1-\alpha \beta}, \|Q_{\min}\|_{\Lip} \leq \frac{L_2(1+\frac 2n)}{1-\beta K_2(1+\frac 2n)} \bigg \}.\]
The ergodicity operator for the quantized game will be 
\begin{align*}
 H_{2,n}&: \mathcal C(n)\times \mathcal P(X_n) \rightarrow P(X_n),
 \\& : (Q,\mu^{\mathcal D})\rightarrow \sum_{x_n \in X_n} p_n(\cdot|x_n,f_n(x_n,Q,\mu^{\mathcal D}),\mu^{\mathcal D})\mu^{\mathcal D}(dx_n) 
\end{align*}
where $f_n(x_n,Q,\mu^{\mathcal D})$ is the minimizer of $H_{1,n}(Q,\mu^{\mathcal D})(x_n,\cdot)$ on $A.$

\begin{definition}
We call the tuple $(X_n,A,p_n,c_n)$ a quantized mean-field game. 
\end{definition}

To show that iterations of the mean-field equilibrium operator  
$$H_n \coloneqq (H_{1,n},H_{2,n})$$ 
converges, we need the one-stage cost function $c_n,$ $\nabla H_{1,n}$ and the transition probability $p_n$ to be Lipschitz continuous, which is the first obstacle we face in quantization of discrete mean-field games on compact continuous state space. Although the Lipschitz coefficient for the variation of states can be obtained for $\nabla H_{1,n}$ and $c_n$ without any obstacle, we will need a further assumption on the transition probability $p_n$ regarding its total variation norm:

\begin{assumption}\label{assumption:b}
\begin{itemize}
\item[ ]
\item[i)] The cost function $c$ is continuous and the transition probability $p$ is weakly continuous. Moreover, we have the following bounds:
\begin{align}
& \| c(\cdot,\cdot,\mu) - c(\cdot,\cdot,\tilde {\mu}) \|_{\infty} \leq L_1 W_1(\mu,\tilde{\mu})
\\ & \sup_{a, \mu} | c(x,a,\mu)-c(\tilde x, a, \mu) | \leq L_2d_X(x,\tilde x)
\\ & \sup_{x \in X} \|p(\cdot|x,a,\mu)-p(\cdot|x,\hat a,\tilde \mu) \|_{TV} \leq K_1 (\|a - \hat a \| +  W_1(\mu,\tilde \mu))
\\ & \sup_{\mu}  \|p(\cdot | x,a,\mu ) - p(\cdot | \hat x,\hat a,\mu )\|_{TV} \leq K_2( d_X(x,y) + \|a-\hat a\|)
\end{align}
\item[ii)] We extend the domain of Assumption \ref{assump1}-b,c as follows. If $\|Q\|_{\infty} \leq \frac{M}{1-\beta}$ then, by abusing the notation on domain of $H^1_1,$ for any $(x,Q,\mu)$, $H^1_1(Q,\mu)(x,\cdot)$ is $\rho$-strongly convex and
\begin{align*}
\sup_{a \in A} \| \nabla H^1_1(Q,\mu)(x,a)- \nabla H^1_1(\tilde Q,\tilde{\mu})(\tilde x,a) \| 
\leq K_F( d_X(x,\tilde x)+\|Q_{\min}-\tilde Q_{\min} \|_{\infty} + W_1(\mu,\tilde{\mu}) ).
\end{align*}
The difference here is that we do not assume that $Q_{\min}$ is Lipschitz continuous, hence we consider the operator $H^1_1$ on an extended domain, say $B_b(X) \times \mathcal P(X),$ where $B_b(X)$ is the space of bounded Borel measurable functions.
\end{itemize}
\end{assumption}

Since $X$ is compact, if $\mu, \nu \in \mathcal P(X),$ we have the following relation between the $1$-Wasserstein metric and the total variation metric over $\mathcal P(X)$ \cite[Theorem 6.15]{Vil09}:
\[
W_1(\mu,\nu)\leq \text{diam}(X) \|\mu-\nu\|_{TV};
\]
that is, the Lipschitz coefficients for the transition probability in Assumption \ref{assump1} can be obtained through the ones in Assumption \ref{assumption:b}.

In case when $\mu$ and $\nu$ are probability measures over $X_n,$ we have the following relation between the total variation metric and the $\ell_1$ metric:
\[
\| \mu- \nu \|_{TV} = \| \mu - \nu \|_{\ell_1} = \sum_{x \in X_n} | \mu(x)-\nu(x)|;
\]
 as the total variation metric is given by
\[
\| \mu-\nu \|_{TV} = 2 \sup_{A} |\mu(A)-\nu(A)|
\]
over all Borel sets $A$ over $X_n.$ This property of the total variation metric in discrete spaces will come handy in the next lemma.

\begin{remark}
In case of general locally compact Polish state spaces, since the diameter of the metric will be no longer bounded, we cannot compare the total variation metric and the $1$-Wasserstein metric as mentioned above thus Assumptions \ref{assump1} and \ref{assumption:b} are no longer related, which means the methods we use here will be ineffective in that case.
\end{remark}

\begin{lemma}\label{lemma:c}
Suppose $n$ is sufficiently large. For any  $x_{n,i},x_{n,j} \in X_n,$ $i\not =j,$ $a,\hat a \in A,$ $Q_1,Q_2,Q_n \in \mathcal C(n)$ and $\mu^{\mathcal D}_1,\mu^{\mathcal D}_2 \in \mathcal P(X_n)$, we have 
\begin{align*}
& \sup_{a,\mu^{\mathcal D}}|c_n(x_{n,i},a,\mu^{\mathcal D})-c_n(x_{n,j},a,\mu^{\mathcal D})| \leq L_2 \bigg(\frac {2}{n}+1 \bigg)d_{X_n}(x_{n,i},x_{n,j})
\\& \sup_{\mu}\| p_n(\cdot |x_{n,i},a,\mu) - p_n(\cdot |x_{n,j},\hat a,\mu)\|_{TV} \leq K_2 \bigg(\frac {2}{n}+1\bigg)d_{X_n}(x_{n,i},x_{n,j})+K_2\|a-\hat a\|
\end{align*}
Furthermore, we also have
\begin{align*}
& \sup_{x,a}|c_n(x,a,\mu^{\mathcal D}_1)-c_n(x,a,\mu^{\mathcal D}_2)| \leq L_1 W_1(\mu^{\mathcal D}_1,\mu^{\mathcal D}_2)
\\& \sup_x \| p_n(\cdot |x,a,\mu^{\mathcal D}_1) - p_n(\cdot |x,\hat a,\mu^{\mathcal D}_2)\|_{TV} \leq K_1 (\|a-\hat a\|+ W_1(\mu^{\mathcal D}_1,\mu^{\mathcal D}_2))
\\& \sup_a \| \nabla H_{1,n}(Q_n,\mu^{\mathcal D}_1)(x_{n,i},a) - \nabla H_{1,n}(Q_n, \mu^{\mathcal D}_2)(x_{n,j},a) \| 
\\&\qquad \qquad  \leq K_F \bigg ( \frac {2}{n} + 1 \bigg) \left( d_{X_n}(x_{n,i},x_{n,j})+W_1(\mu^{\mathcal D}_1,\mu^{\mathcal D}_2) + \|Q_{1,\min}-Q_{2,\min} \|_{\infty} \right).
\end{align*}
\end{lemma}
\begin{proof}
For the bound on the cost function, using Jensen's inequality and triangle inequality, we obtain
\begin{align*}
|c_n(x_{n,i},a,\mu^{\mathcal D})-c_n(x_{n,j},a,\mu^{\mathcal D})|  &  = \bigg | \int_{S_{n,i}} c(z,a,\mu^{\mathcal C})\nu_{n,i}(dz) -\int_{S_{n,j}} c(r,a,\mu^{\mathcal C})\nu_{n,j}(dr) \bigg |
\\& \qquad = \bigg | \int_{S_{n,j}}\int_{S_{n,i}} c(z,a,\mu^{\mathcal C})-c(r,a,\mu^{\mathcal C})\nu_{n,i}(dz)\nu_{n,j}(dr) \bigg |
\\& \qquad \leq \int_{S_{n,j}}\int_{S_{n,i}} | c(z,a,\mu^{\mathcal C})-c(x_{n,i},a,\mu^{\mathcal C}) | \nu_{n,i}(dz)\nu_{n,j}(dr)
\\& \qquad \qquad + \int_{S_{n,j}}\int_{S_{n,i}} | c(x_{n,i},a,\mu^{\mathcal C})-c(x_{n,j},a,\mu^{\mathcal C}) | \nu_{n,i}(dz)\nu_{n,j}(dr)
\\& \qquad \qquad + \int_{S_{n,j}}\int_{S_{n,i}} | c(x_{n,j},a,\mu^{\mathcal C})-c(r,a,\mu^{\mathcal C}) | \nu_{n,i}(dz)\nu_{n,j}(dr)
\\& \qquad \leq \int_{S_{n,j}}\int_{S_{n,i}} L_2 d_X(z,x_{n,i}) \nu_{n,i}(dz)\nu_{n,j}(dr)
\\& \qquad \qquad + \int_{S_{n,j}}\int_{S_{n,i}} L_2 d_X(x_{n,i},x_{n,j}) \nu_{n,i}(dz)\nu_{n,j}(dr)
\\& \qquad \qquad + \int_{S_{n,j}}\int_{S_{n,i}} L_2 d_X(x_{n,j},r) \nu_{n,i}(dz)\nu_{n,j}(dr)
\\& \qquad \leq \frac{2L_2}{n} + L_2d_X(x_{n,i},x_{n,j})
\\& \qquad \leq L_2\bigg (\frac{2}{n} + 1\bigg )d_{X_n}(x_{n,i},x_{n,j})
\end{align*}
by using the fact that $d_X \leq 1$ and we have discrete metric on $X_n.$

Note that total variation distance is the same as $\ell_1$ on probability measures. 
Using this and a straightforward application of triangle inequality as in the case of the cost function gives us the following Lipschitz bound for $p_n:$
\begin{align*}
&\|p_n(\cdot|x_{n,i},a,\mu^{\mathcal D})-p_n(\cdot|x_{n,j},a,\mu^{\mathcal D}) \|_{TV} 
\\& = \sum_{p=1}^{k_n} |p_n(x_{n,p}|x_{n,i},a,\mu^{\mathcal D})-p_n(x_{n,p}|x_{n,j},a,\mu^{\mathcal D} |
\\& = \sum_{p=1}^{k_n} \bigg |\int_{S_{n,i}} p(S_{n,p}|z,a,\mu^{\mathcal C})\nu_{n,i}(dz) - \int_{S_{n,j}}p(S_{n,p}|r,a,\mu^{\mathcal C})\nu_{n,j}(dr) \bigg |
\\& \leq \sum_{p=1}^{k_n} \int_{S_{n,i}}\int_{S_{n,j}} \bigg | p(S_{n,p}|z,a,\mu^{\mathcal C})-p(S_{n,p}|x_{n,i},a,\mu^{\mathcal C}) \bigg | \nu_{n,i}(dz) \nu_{n,j}(dr)
\\& \qquad + \sum_{p=1}^{k_n} \int_{S_{n,i}}\int_{S_{n,j}} \bigg |p(S_{n,p}|x_{n,i},a,\mu^{\mathcal C})-p(S_{n,p}|x_{n,j},a,\mu^{\mathcal C})\bigg |\nu_{n,i}(dz) \nu_{n,j}(dr)
\\& \qquad + \sum_{p=1}^{k_n} \int_{S_{n,i}}\int_{S_{n,j}} \bigg | p(S_{n,p}|x_{n,j},a,\mu^{\mathcal C})-p(S_{n,p}|r,a,\mu^{\mathcal C}) \bigg | \nu_{n,i}(dz) \nu_{n,j}(dr)
\\&\leq \int_{S_{n,i}}\int_{S_{n,j}} \| p(\cdot|z,a,\mu^{\mathcal C}) - p(\cdot | x_{n,i},a,\mu^{\mathcal C}) \|_{TV} \nu_{n,i}(dz) \nu_{n,j}(dr)
\\& \qquad + \int_{S_{n,i}}\int_{S_{n,j}}  \| p(\cdot|x_{n,i},a,\mu^{\mathcal C}) - p(\cdot | x_{n,j},a,\mu^{\mathcal C}) \|_{TV} \nu_{n,i}(dz) \nu_{n,j}(dr)
\\& \qquad + \int_{S_{n,i}}\int_{S_{n,j}}  \| p(\cdot|z,a,\mu^{\mathcal C}) - p(\cdot | x_{n,i},a,\mu^{\mathcal C}) \|_{TV} \nu_{n,i}(dz) \nu_{n,j}(dr)
\\& \leq \frac {2K_2}{n} + K_2d_X(x_{n,i},x_{n,j})
\\& \leq \bigg ( \frac {2K_2}n + K_2 \bigg ) d_{X_n}(x_{n,i},x_{n,j}).
\end{align*}
In the above argument, we only consider the Lipschitz bound on $p_n$ when the state variables are different. One can extend the same argument to the case where both the state variables and action variables are different. In this case, we get an extra $K_2 \|a-\tilde a\|$ term. Moreover, we can obtain other Lipschitz bounds on $c_n$ and $p_n$ similarly by noting the following fact. To obtain the desired inequalities for variations on $\mathcal P(X_n)$, we will use the relation between $\mu^{\mathcal D}$ and $\mu^{\mathcal C}$ as in (\ref{disctocontprob}). By definition, each such $\mu^{\mathcal C}$ obtained from $\mu^{\mathcal D}$ is atomic and completely determined by $\{S_{n,i}\}_{i=1}^{k_n}.$ So to overcome the extra difficulty obtained for variations on $\mathcal P(X_n),$ observe that for atomic measures $\mu^{\mathcal C}_1$ and $\mu^{\mathcal C}_2$ on $\mathcal P(X)$ that is obtained from a discrete probability measure on $X_n$ that is completely determined by the partition $(S_{n,i}),$ we have
\begin{align*}
W_1(\mu^{\mathcal C}_1,\mu^{\mathcal C}_2)
& = \sup_{\|g\|_{\Lip,d_X} \leq 1} \bigg | \int_X g(t)\big ( \mu^{\mathcal C}_1(dt)-\mu^{\mathcal C}_{2}(dt) \big ) \bigg |
\\& = \sup_{\|g\|_{\Lip,d_X} \leq 1} \bigg | \sum _{i=1}^{k_n} g(x_{n,i})\big(\mu^{\mathcal C}_1(S_{n,i})-\mu^{\mathcal C}_2(S_{n,i})\big) \bigg |
\\& = \sup_{\|g\|_{\Lip,d_X} \leq 1} \bigg | \sum _{i=1}^{k_n} g(x_{n,i})\big(\mu^{\mathcal D}_1(x_{n,i})-\mu^{\mathcal D}_2(x_{n,i})\big) \bigg |
\\& \leq \sup_{\|g\|_{\Lip,d_{X_n}} \leq 1} \bigg | \sum _{i=1}^{k_n} g(x_{n,i})\big(\mu^{\mathcal D}_1(x_{n,i})-\mu^{\mathcal D}_2(x_{n,i}) \big) \bigg |
\\& = W_1(\mu^{\mathcal D}_1,\mu^{\mathcal D}_2)
\end{align*}
where $\|\cdot \|_{\Lip,d_X}$ (resp. $\|\cdot\|_{\Lip,d_{X_n}}$) gives the Lipschitz constant with respect to the metric $d_X$ (resp. $d_{X_n}$) on $X$ (resp. on $X_n$), and using the fact that since $d_X \leq d_{X_n}$ on ${X_n \times X_n}$, any Lipshcitz function with Lipschitz constant less than $1$ with respect to the metric $d_X$ can be restricted to $X_n$ and will have Lipschitz coefficient less than $1.$  Therefore, we can obtain Lipschitz coefficients for the state-measures on $X_n$ from total variation norm of the induced measures on the continuous state space measures.

We shall obtain the Lipschitz bound on the operator $\nabla H_{1,n}$ as follows. We can extend any $Q \in \mathcal C(n)$ to $B_b(X)$ as $Q\circ(T_n,I)$ where $(T_n,I)(x,a)=(T_n(x),a),$ which is measurable. Since \[
\|Q\circ (T_n,I) \|_{\infty } =\sup_{(x,a) \in X \times A} |Q(T_n(x),a)| \leq \sup_{(x_n,a) \in X_n \times A}|Q(x_n,a)|
\] 
we have $\| Q \circ (T_n,I) \|_{\infty} \leq \frac M {1-\beta}$ hence extension of elements in $\mathcal C(n)$ preserves the bound under uniform norm when extended to $X \times A$ as we just described. By abusing the notation on the domain of $Q$ functions for $H_{1}^1$ as in Assumption \ref{assumption:b}-ii, we have
\[
\nabla H_{1,n}(Q,\mu^{\mathcal D})(x_{n,i},a) = \int_{S_{n,i}} \nabla H_1^1(Q\circ (T_n,I),\mu^{\mathcal C})(x,a) \nu_{n,i}(dx).
\]
Note that Assumption \ref{assump1}-c is not applicable here since $Q\circ (T_n,I)$ does not satisfy the conditions. Hence the bound $K_F$ has to be obtained through Assumption \ref{assumption:b}. With this observation, similar to the case of the cost function, we can obtain the following Lipschitz bound on variation over states as the bound on $Q\circ (T_n,I)$ is identical to those in the domain $\mathcal C(n):$
\[
\| \nabla H_{1,n}(Q,\mu^{\mathcal D})(x_{n,i},a)-\nabla H_{1,n}(Q,\mu^{\mathcal D})(x_{n,j},a) \| \leq K_F\bigg ( \frac 2n +1 \bigg )d_{X_n}(x_{n,i},x_{n,j}).
\]
This concludes the inequality we wanted for variation on the states. 
 
In case of variation over $Q$ functions for $\nabla H_{1,n}$, arguing exactly as above, the variation over $Q$ functions for $Q_1,Q_2 \in \mathcal C(n)$ will be of the form
\begin{align*}
\| \nabla H_{1,n}(Q_1,\mu^{\mathcal D})(x,a)-\nabla H_{1,n}(Q_2,\mu^{\mathcal D})(x,a) \| 
&\leq K_F(1+2/n)\|Q_{1,\min}\circ T_n -Q_{2,\min}\circ T_n \|_{\infty}
\\&=  K_F(1+2/n)\|Q_{1,\min} -Q_{2,\min}\|_{\infty}
\end{align*}
Inequalities for variations on $\P(X_n)$ can be obtained in a similar fashion given the discussion above of Wasserstein distances on continuous and discrete state spaces.
\end{proof}

\begin{remark}\label{rem:a}
The scaling trick that we have done on the metric $d_X$ is necessary and affect the choice we can have for $d_{X_n}$. Indeed, if $X=[0,1]$ and we pick the representatives for the quantization of $X$ such that minimum distance between them is $1/n^2,$ then as in the proof of Lemma \ref{lemma:c}, we can only obtain
\begin{align*}
|c_n(x_{n,i},a,\mu^{\mathcal D})-c_n(x_{n,j},a,\mu^{\mathcal D})| 
& \leq  \frac {2L_2}{n}+L_2d_X(x_{n,i},x_{n,j})
\\ &   \leq  L_2 \bigg (\frac 2{n\min_{i \not = j}d_X(x_{n,i},x_{n,j})}+1 \bigg )d_X(x_{n,i},x_{n,j}).
\end{align*}
Thus, if $d_{X_n}=d_X|_{X_n \times X_n},$ we see that 
\[
|c_n(x_{n,i},a,\mu^{\mathcal D})-c_n(x_{n,j},a,\mu^{\mathcal D})| \leq K_nd_{X_n}(x_{n,i},x_{n,j})
\]
for $K_n = L_2(2n+1)$ as a consequence of the inequality above. Similar results also hold if $d_{X_n}$ is a multiple of $d_X$ too. Thus, $d_{X_n}$ cannot be chosen as a multiple of $d_X$ and has to be greater than $1$ when considering the distance between two distinct elements. This shows the necessicity of the scaling argument we did at the beginning of this section.

In case $d_X \not \leq 1$ but $d_{X_n}(x_{n,i},x_{n,j}) = 1_{j \not = i},$ the $1-$Wasserstein distance on $X_n$ and $X$ are no longer related, this also creates an obstacle in itself as we will only have $W_1(\mu^{\mathcal C}_1,\mu^{\mathcal C}_2) \leq M W_1(\mu^{\mathcal D}_1,\mu^{\mathcal D}_2),$ which can break the condition in Theorem \ref{assump1} if $M$ is large enough, as we will then need to have $MK_2<1$ for the corresponding $k_2$ term in Theorem \ref{thrm1}.
\end{remark}

As a consequence of Lemma \ref{lemma:c}, we see that the operator $H_{1,n}$ is well defined similar to the continuous state space case.

\begin{proposition}\label{prop:a}
For sufficiently large $n$, the operator $H_{1,n}$ is well defined in the sense that for any $\mu ^{\mathcal D}\in \mathcal P(X_n)$ and $Q \in \mathcal C(n)$, we have $H_{1,n}(Q,\mu ^{\mathcal D}) \in \mathcal C(n).$
\end{proposition}

\begin{proof}
Since $d_{X_n}\leq 1$, we have the relation $W_1 \leq \|\cdot \|_{TV}$ between the $1$-Wasserstein metric and the total variation distance on $\mathcal P(X_n).$ The very same proof in \cite[Lemma 1]{anahtarci2020value} concludes the statement.
\end{proof}

Note that it is straightforward to prove that $H_{1,n}(Q,\mu)(x,\cdot)$ is $\rho$-strongly convex for any $(Q,\mu,x)\in \mathcal C(n) \times \mathcal P(X_n)\times X_n$. Hence, if the Lipschitz bounds for quantized MFG satisfies the conditions in Theorem \ref{thrm1}, then there exists a unique fixed point of the operator $H_n = (H_{1,n},H_{2,n})$ on $\mathcal C(n) \times \mathcal P(X_n)$ whenever $n$ is sufficiently large.

\begin{lemma}\label{lemma:e}
If $n$ is sufficiently large and Theorem \ref{thrm1} is satisfied for the nominal game, then the quantized game also have a unique mean-field equilibrium that is obtained through value iteration algorithm provided that Assumption \ref{assumption:b} holds.
\end{lemma}

\begin{proof}
Whenever $n$ is large enough we have
\begin{align*}
k_{1,n} = \frac{K_FK_1\left(1+\frac 2n\right )}{\rho(1-\beta)} < k_{2,n}=\frac {1-\bigg (K_2\bigg(1+\frac{2}{n}\bigg)+K_1\bigg)(\big (\frac 2n +1 \big )\frac{K_F}{\rho}+1)}{L_1+\beta \big(\frac 2n+1\big )\frac{L_2K_1}{1-\beta K_2(2/n+1)}},
\end{align*}
and 
\[ k_n = \max \left ( \beta , \left ( \frac{K_F}{\rho}\left(\frac 2n +1\right) +1 \right ) \left(K_2 \left(\frac 2n +1 \right)+K_1 \right ) \right)< 1 \]
as a consequence of Lemma \ref{lemma:c}. Thus, we can apply Theorem $1$ to the quantized games, as the operators $H_{n}$, when $n$ is sufficiently large, are well defined by Proposition \ref{prop:a}.
\end{proof}

\subsection{Extension of Quantized MFGs to  Continuous State Space}\label{e-mfgs}

In the previous section, we discretized the state space and found mean equilibria for quantized games using the value iteration algorithm. To utilize these mean-field equilibria to approximate the nominal model's equilibrium, we must extend the fixed points we found to the original state space, denoted as $X$. In this subsection, we will demonstrate how to extend the quantized mean-field games to the state space $X$.

We will further establish that the convergence of iterations in the quantized mean-field games implies the convergence of iterations in the extended game. This is not an immediate conclusion because the $Q$-functions obtained through the iterations of the extended games will not be continuous on the state space $X$, and so, the results established before about the convergence of value iteration for continuous state games cannot be applied directly to the extended game (i.e. Theorem~\ref{thrm1}).

Since our aim is to use robustness result established in the previous section for the approximation of continuous state-space mean-field games, we first have to extend our quantized game to the full space $X.$ To this end, we extend $c_n$ and $p_n$ to $X \times A \times \mathcal P(X)$ by using the relations	

\[ 
\hat c_n(x,a,\mu) = \int_{S_{n,i(T_n(x))}} c(y,a,\mu) \nu_{n,i(T_n(x))}(dy), \text{ and} \]

\[ 
\hat p_n(\cdot | x,a,\mu) = \int_{S_{n,i(T_n(x))}} \sum_{p=1}^{k_n} \delta_{x_{n,p}}(\cdot) p(S_{n,p} | y,a,\mu) \nu_{n,i(T_n(x))}(dy)
\]
where $i(T_n(x))$ is the $i \in \{1,\cdots,k_n\}$ such that $T_n(x) \in S_{n,i}.$ Similar to the transformation between a probability measure on the quantized state space and a probability measure on the continuous state space, between $p_n$ and $\hat p_n$ we have the following relations
\begin{equation}\label{transform:transitionprob}
\hat p_n(\cdot | x,a,\mu ^{\mathcal C}) = \sum_{i=1}^{k_n} \delta_{x_{n,i}}(\cdot) p_n(x_{n,i}|T_n(x),a,\mu^{\mathcal D})
\end{equation}
and 
\begin{equation}
p_n(\cdot|T_n(x),a,\mu^{\mathcal D}) = \sum_{i=1}^{k_n} \delta _{x_{n,i}}(\cdot) \hat p_n(S_{n,i}|x,a,\mu^{\mathcal C}),
\end{equation}
whenever $\mu^{\mathcal C}$ is obtained from $\mu^{\mathcal D}$ as described before.

By using $\hat c_n$ and $\hat p_n$ we can extend the operator $H_{1,n}$ of the quantized game to $B^{\ell}(X) \times \mathcal P(X)$ as follows
\begin{align*}
 \hat {H}_{1,n} &:B^{\ell}(X\times A) \times \mathcal P(X)  \rightarrow B^{\ell}(X\times A) ,
 \\&:(Q,\mu^{\mathcal C}) \rightarrow \hat c_n(x,a,\mu^{\mathcal C}) + \beta \int_X Q_{\min}(y)\hat p_n(dy|x,a,\mu^{\mathcal C})
\end{align*}
where $B^{\ell}(X)$ is the set of all bounded lower-semi analytic functions unlike to the quantized game, where we considered domain $\mathcal C(n).$ Recall that for any $(x,a) \in X \times A,$ we define the map  $(T_n,I)$ on $X \times A$ with the relation $(T_n,I)(x,a)=(T_n(x),a)$. Then, for any $\mu^{\mathcal D} \in \mathcal P(X_n)$ and $Q \in \mathcal C(n)$, we have
\begin{equation}\label{quantextend}
\hat H_{1,n}(Q\circ(T_n,I),\mu^{\mathcal C})(x,a) = H_{1,n}(Q,\mu^{\mathcal D})(T_n(x),a).
\end{equation}

Similarly we can define the ergodicity operator for the extended game as 
\begin{align*}
\hat H_{2,n}(Q,\mu^{\mathcal C}) = \int_{X}\int_A \hat p_n(\cdot|x,a,\mu^{\mathcal C}) \delta_{\underset{a' \in A}{\mathrm{argmin}}\hat H_{1,n}(x,Q,\mu^{\mathcal C},a')}(da)\mu^{\mathcal C}(dx)
\end{align*}
on the domain $B^{\ell}(X) \times \mathcal P(X).$ If we define 
\[ 
\hat f_n(x,Q,\mu^{\mathcal C}) = \underset{a' \in A}{\mathrm{argmin}}\hat F_n(x,Q,\mu^{\mathcal C},a'),
\]
then for $\mu^{\mathcal C} \in \mathcal P(X)$ that is obtained from $\mu^{\mathcal D} \in \mathcal P(X_n)$ through (\ref{disctocontprob}) and  for any $(x,Q) \in X \times \mathcal C(n)$, we get
\[
f_n(T_n(x),Q,\mu^{\mathcal D}) = \hat f_n(x,Q\circ(T_n,I),\mu^{\mathcal C})
\]
as a consequence of (\ref{quantextend}).

\begin{definition}
We will call the tuple $(X,A,\hat p_n, \hat c_n)$ an extended (mean-field) game obtained from a quantized model. 
\end{definition}

Following lemma demonstrates us that the operator $\hat H_{1,n}$ is well defined. The proof is standard in the literature of Markov decision processes, however for completeness we will include a complete proof.

\begin{lemma}\label{lemma:aa}
For fixed $\mu \in \mathcal P(X),$ the map $ \bullet \rightarrow \hat H_{1,n}(\bullet,\mu)$ maps $B^{\ell}(X\times A)$ to $B^{\ell}(X\times A).$
\end{lemma}

\begin{proof}
Let $Q$ be a bounded lower-semi analytic function on $X \times A.$ Then $\min_a Q(x,a)$ is also lower semi-analytic by \cite[Proposition 7.47.]{bertsekas1996stochastic}. Since $\min_a Q(x,a)$ is lower semi-analytic, by \cite[Proposition 7.48]{bertsekas1996stochastic}, for any $(x,a,\mu ^{\mathcal C}) \in X \times A \times \mathcal P(X)$, we have that 
\[
\int_X \sum_{p=1}^{k_n} Q_{\min}(t)\delta_{x_{n,p}}(dt)p(S_{n,p}|x,a,\mu ^{\mathcal C})
\] is lower semi-analytic. By  \cite[Lemma 30-(4) pg. 178]{bertsekas1996stochastic}, we have that 
\[ 
c(x,a,\mu)+\beta \int_X \sum_{p=1}^{k_n} Q_{\min}(t) \delta_{x_{n,p}}(dt) p(S_{n,p}|x,a,\mu) 
\] is lower semi-analytic. For the stochastic kernel $\tilde \nu_n(dx|y) = \sum_{i=1}^{k_n} \nu _{n,i}(dx)1_{S_{n,i}}(y)$, by \cite[Proposition 7.48]{bertsekas1996stochastic} we obtain that
\[ 
\int_X\bigg ( c(T_n(x),a,\mu^{\mathcal C})+\beta \int_X \sum_{p=1}^{k_n} Q_{\min}(t) \delta_{x_{n,p}}(dt) p(S_{n,p}|x,a,\mu^{\mathcal C})\bigg )\tilde \nu (dx|y) 
\]
is lower semi-analytic. However, the quantity above is equal to
\[ \hat H_{1,n}(Q,\mu)(x,a)=\hat c_n(x,a,\mu)+\beta \int_{X} Q_{\min}(t)\hat p_n(dt|x,a,\mu), \]
thus the result follows.
\end{proof}

Since $\hat H_{1,n}$ closed under the space of lower semi-analytic functions, there exists a lower semi-analytic minimizer for it \cite{shreve1979universally}. Thus, since lower semi-analytic selections are universally measurable, the map $\hat H_n \coloneqq (\hat H_{1,n},\hat H_{2,n})$ is also well defined.

For any $\mu \in \mathcal P(X),$ it is known that $\hat H_{1,n}$ is a contraction on $B^{\ell}(X)$ with modulus $\beta<1 .$ Furthermore, limit of lower semi-analytic maps are also lower semi-analytic, thus the optimal $Q$ function for $\hat H_{1,n}$ under $\mu$ must be lower semi-analytic, which exists by the Banach fixed point theorem. More can be said about the fixed $Q$-function for the extended game due to its relation with the quantized fixed $Q$-function of quantized game $H_{1,n}.$

\begin{remark}
The transition probability of $p$ is only assumed to be weakly continuous in Assumption \ref{assumption:b}, and thus, we cannot guarantee that the map $\hat H_{1,n}$ maps $B(X)$ into itself \cite[pg 9.]{saldi2017asymptotic}. This observation motivates our the reason behind the rather large domain of $Q$ functions, namely $B^{\ell}(X),$ that the extended games $\hat H_{n}$ are defined over.
\end{remark}

Now we will describe how to relate the iterations of the extended game to that of the quantized games. Since we can guarantee Borel measurable $Q$-functions on the iterations of the quantized games, this scheme will also allow us to guarantee the same for the extended games under the same conditions.

For the $n^{th}$ quantized game, under a probability measure $\mu^{\mathcal D} \in \mathcal P(X_n),$ when $\beta <1$, there exists a unique optimal $Q_{\mu^{\mathcal D}}$ function. Extending this value function to $X$ via $\hat Q_{\mu^{\mathcal C}}=Q_{\mu^{\mathcal D}}\circ (T_n,I),$ we see that $\hat Q_{\mu^{\mathcal C}}$ is the optimal $Q$ function for the extended game under the probability measure $\mu^{\mathcal C}.$ Since $Q_{\mu^{\mathcal D}}$ and $(T_n,I)$ are Borel measurable on their respective domains, we obtain that $\hat Q_{\mu^{\mathcal C}}$ is Borel measurable. Uniqueness of the $Q$ function for $\hat H_{1,n}$  under the probability measure $\mu^{\mathcal C}$ implies that  the optimal $Q$ function for $\hat H_{1,n}$ under $\mu ^{\mathcal C}$ must be Borel measurable. This also implies that the corresponding ergodicity operator is also defined via Borel measurable policy. This argument will also give us a method to find mean-field equilibrium for the extended game. 

The above observation will allow us to obtain a mean-field equilibrium with a Borel measurable optimal for the extended mean-field games rather than a mean-field equilibrium with a lower-semi analytic optimal policy through value iteration algorithm. Since lower-semi analytic policies are not necessarily Borel measurable, for convenience, we will restrict ourselves to the Borel measurable ones. 

\begin{definition}\label{extend:mfe}
We will call a Borel (measurable) policy $\pi$ and a state-measure $\mu$ on $X$ a mean-field equilibrium for the extended game, if $\pi$ is optimal for the optimal $Q$ function$, Q_{\mu},$ of $\hat H_{1,n}$ under $\mu$ and $\mu$ is a fixed point of the ergodicity operator $\hat H_{2,n}(Q_{\mu},\mu).$
\end{definition}

Following proposition shows that as long as we can find a mean-field equilibrium for the quantized game $H_n,$ we can find a mean-field equilibrium with Borel measurable optimal policy to the extended game $\hat H_n$ in the sense of Definition \ref{extend:mfe} above. Furthermore, it also demonstrates that the mean-field equilibria described in Definition \ref{extend:mfe} are precisely the ones obtainable from the quantized games when we extend the domain of $Q$-functions for the quantized games to the space $B(X_n)=C(X_n)$ from $\mathcal C(n)$.

\begin{proposition}\label{prop:c}
Let $(\pi,\mu^{\mathcal D})$ be a mean-field equilibrium for the quantized game $H_n.$ Then $(\pi(\cdot|T_n(\cdot)),\mu^{\mathcal C}),$ where $\mu^{\mathcal C}$ is obtained from $\mu^{\mathcal D}$ through the Equation (\ref{disctocontprob}), is a mean-field equilibrium for the extended game $\hat H_n.$ Similarly, if $(\pi,\mu^{\mathcal C})$ is a mean-field equilibrium for the extended game $\hat H_n$ in the sense of Definition \ref{extend:mfe}, then $(\pi, \mu^{\mathcal D})$ is a mean-field equilibrium for the quantized game $H_n,$ where $\mu^{\mathcal D}$ is obtained from $\mu^{\mathcal C}$ through the Equation (\ref{conttodiscprob}).
\end{proposition}
\begin{proof}
Let $(\pi,\mu ^{\mathcal D})$ be a mean-field equilibrium for the quantized game. We will show that $\pi$ induces an optimal policy for the extended game under $\mu ^{\mathcal C}$ first.

Define the measure $\nu(x,da)=\pi(da|x)\mu ^{\mathcal D}(x)$ and note also that we have
$$H_{1,n}(Q_{\mu^{\mathcal D} },\mu ^{\mathcal D})=Q_{\mu^{\mathcal D}}.$$ 
Since disintegration of $\nu$ gives the mean-field equilibrium, for the following set on $X_n \times A$ 
\begin{align*}
\tilde C = \bigg \{ &(x_n,a) \in X_n \times A :  H_{1,n}(Q_{\mu^{\mathcal D}},\mu ^{\mathcal D})(x_n,a) = Q_{\mu^{\mathcal D},\min}(x_n) \bigg \}
\end{align*}
we have \( \nu (\tilde C) =1 \) as a consequence of \cite[Theorem 3.6]{SaBaRaSIAM}; that is, $\nu$ concentrates on the set of optimal state-action pairs.
We define measurable set $\hat C$ analogous to $\tilde C$ on $X \times A$ as follows
\[
\hat C=\bigg \{  (x,a) \in X\times A: \hat H_{1,n}(Q_{\mu^{\mathcal D}}\circ(T_n,I),\mu^{\mathcal C})(x,a) 
= Q_{\mu ^{\mathcal D},\min}\circ (T_n(x)) \bigg \}.
\]
For the measure
$$\nu^{\mathcal C}(dx,da) = \pi(da|T_n(x))\mu^{\mathcal C}(dx)$$
in $\mathcal P(X \times A),$ we obtain the following
\begin{align*}
\nu^{\mathcal C} (\hat C) 
&= \int_{X \times A} 1_{\hat C}(x,a)\nu^{\mathcal C}(dx,da) 
\\&= \int_X \int_A 1_{\hat C}(x,a) \pi(da|T_n(x))\mu^{\mathcal C}(dx)
\\&= \int_X \int_A 1_{\hat C}(x,a) \pi(da|T_n(x))\sum_{i=1}^{k_n} \mu ^{\mathcal D}(x_{n,i})\delta_{x_{n,i}}(dx)
\\&= \sum_{i=1}^{k_n} \int_{A} 1_{\hat C}(x_{n,i},a)\pi(da|x_{n,i})\mu ^{\mathcal D}(x_{n,i})
\\&= \sum_{i=1}^{k_n} \int_{A} 1_{(T_n,I)(\hat C)}(x_{n,i},a)\pi(da|x_{n,i})\mu ^{\mathcal D}(x_{n,i})
\\&= \int_{X_n \times A} 1_{(T_n,I)(\hat C)}(x,a)\nu(dx,da)
\\&= \nu((T_n,I)(\hat C))
\\& = \nu(\tilde C)
\\& =1
\end{align*}
by the Equation (\ref{transform:transitionprob}) and the relation $(T_n,I)(\hat C)=\tilde C ,$ which is a direct consequence of the Equation (\ref{quantextend}). Therefore, the conditional distribution obtained from the disintegration of the measure $\nu^{\mathcal C},$ $\pi(\cdot | T_n(x)),$ is optimal for the extended game under $\mu^{\mathcal C}$; that is, under this policy, $\mu^{\mathcal C}$ concentrates on optimal state-action pair.

It remains to check that $\mu^{\mathcal C}$ satisfies the ergodicity condition. Since $\mu ^{\mathcal D}$ satisfies the ergodicity condition
\[
\mu^{\mathcal D}(\cdot) = \int_{X_n \times A} p_n(\cdot | x_n,f_n(x_n,Q_{\mu^{\mathcal D}},\mu ^{\mathcal D}),\mu ^{\mathcal D})\pi(da|x_n)\mu^{\mathcal D}(dx_n) ,
\]
using the relation between $\mu^{\mathcal C}$ and $\mu^{\mathcal D}$, we obtain
\begin{align*}
\mu^{\mathcal C}(\bullet) 
&= \sum_{i=1}^{k_n} \mu^{\mathcal D}(x_{n,i})\delta_{x_{n,i}}(\bullet)
\\& = \sum_{i=1}^{k_n} \delta_{x_{n,i}}(\bullet) \sum_{j=1}^{k_n} \int_A p_n(x_{n,i}|x_{n,j},a,\mu^{\mathcal D}),\mu^{\mathcal D})\pi(da|x_{n,j})\mu^{\mathcal D}(dx_{n,j})
\\&=\sum_{j=1}^{k_n} \int_A \sum_{i=1}^{k_n} \delta_{x_{n,i}}(\bullet)p_n(x_{n,i}|x_{n,j},a,\mu^{\mathcal D}),\mu ^{\mathcal D})\pi(da|x_{n,j})\mu^{\mathcal D}(dx_{n,j})
\\&=\int_X \int_A \sum_{i=1}^{k_n} \delta_{x_{n,i}}(\bullet) p_n(x_{n,i}|T_n(x),a,\mu^{\mathcal D}),\mu ^{\mathcal D})\pi(da|T_n(x))\mu^{\mathcal C}(dx)
\\&=\int_{X \times A} \hat p_n(\bullet|x,a,\mu ^{\mathcal C})\nu^{\mathcal C}(dx,da).
\end{align*}
Hence, $(\pi,\mu ^{\mathcal D})$ induces a mean-field equilibrium for the extended game. Going backwards in the equalities above gives that mean field equilibrium of the extended game induces a mean-field equilibrium for the quantized game as well.
\end{proof}

We can further obtain that uniqueness of the mean-field equilibrium for the extended game $\hat H_n,$ in the sense of Definition \ref{extend:mfe}, is equivalent to uniqueness of the mean-field equilibrium for the quantized game $H_n$. Since the $Q$-functions we obtain through the iterations of the operators $\hat H_n$ are not Lipschitz continuous on the state space, as these $Q$-functions are constant on each $S_{n,p},$ we cannot use our existing results; we will need this uniqueness result for the convergence of the value-iteration algorithm on the extended mean-field game.

\begin{proposition}\label{prop:3}
The mean-field equilibrium for a quantized game is unique if and only if so is the one for the associated extended game in the sense of Definition \ref{extend:mfe}.
\end{proposition}

\begin{proof}

To show that uniqueness of the MFE of the quantized game implies that of the extended game, we first observe that the state measures of the MFE for the extended game are also atomic. Indeed, since $\hat H_{1,n}(\cdot,\cdot)(x,\cdot) = \hat H_{1,n}(\cdot,\cdot)(y,\cdot)$ for any $x,y \in S_{n,i},$ optimal policy is constant on each partition $S_{n,i}.$ Thus, if $(\hat \pi, \hat \mu^{\mathcal C})$ is a mean field equilibrium for the extended game, due to the invariance condition, $\hat \mu^{\mathcal C}$ is atomic and completely determined by its values on the partition $(S_{n,i})_{i=1}^{k_n}.$ Thus, if the quantized game has a unique mean field equilibrium, any two mean field equilibria $(\hat \pi_1,\hat \mu^{\mathcal C}_1)$ and $(\hat \pi_2,\hat \mu^{\mathcal C}_2)$ of the extended game induce the same discrete state-measures $\hat \mu^{\mathcal D}_1=\hat \mu^{\mathcal D}_2$ that satisfies (\ref{conttodiscprob}). In turn, this means that $\hat \mu^{\mathcal C}_1=\hat \mu^{\mathcal C}_2$ and consequently the optimal policies must also be the same due to the uniqueness of optimal $Q$-functions. Therefore, there is a unique mean-field equilibrium for the extended game.

The reverse implication can be obtained in a similar fashion.
\end{proof}

We are now in position to show that iterations of the extended games converge under the same conditions we have imposed on the quantized games in previous section. Main differences here will be, compared to the quantized games, that the starting point of the algorithm should be an element that we obtain from the quantized model and the convergence of the iterations is not obtained under contraction on the extended space but through the convergence of the iterations of the quantized games.

\begin{proposition}\label{prop:d}
Suppose $k_1<k_2$ and $k<1$ holds. Under Assumption \ref{assumption:b}, assuming that $n$ is sufficiently large, any iteration of the operator $\hat H_n$ that starts from an initial datum $(Q  \circ (T_n,I),\mu^{\mathcal C}) \in B(X\times A) \times \mathcal P(X),$ where $Q \in \mathcal C(n)$ and $\mu^{\mathcal C}$ is obtained from some $\mu^{\mathcal D} \in \mathcal P(X_n)$  as described before, converges to a fixed point in $((B(X \times A),\|\cdot\|_{\infty}),(\mathcal P(X),W_1))$ where $\| \cdot \|_{\infty}$ is the uniform norm over $B(X \times A).$
\end{proposition}
\begin{proof}
Let 
\[
\hat {\mathcal C}(n) = \{ Q\circ (T_n,I) \in B(X\times A): Q \in \mathcal C(n) \}.
\] Since $Q \in \mathcal C(n)$ is Borel measurable, for any $Q \in \mathcal C(n)$, we have $Q \circ (T_n,I) \in \hat {\mathcal C}(n)$. By $ \hat H_n |_{\hat {\mathcal C}(n)\times \mathcal P(X)}$ we denote the restriction of $\hat H_n$ to the domain $\hat{\mathcal C}(n) \times \mathcal P(X) \subset B(X\times A) \times \mathcal P(X).$
Then, for $\mu^{\mathcal C} \in \mathcal P(X)$ that is obtained from some $\mu^{\mathcal D} \in \mathcal P(X_n)$ as in (\ref{disctocontprob}), we have 
\[ \hat H_{1,n} |_{\hat { \mathcal C}(n)\times \mathcal P(X)} (Q,\mu^{\mathcal C}) = H_{1,n}(\tilde Q,\mu^{\mathcal D})\circ(T_n,I) \]
and \[ \hat H_{2,n} |_{\hat {\mathcal C}(n)\times \mathcal P(X)}(Q,\mu^{\mathcal C}) = H_{2,n}(\tilde Q,\mu^{\mathcal D}), \] 
where $\tilde Q \circ (T_n,I) = Q$. By applying Proposition \ref{prop:3} on $\hat H_n |_{\hat {\mathcal C}(n)\times \mathcal P(X)}$, we see that there exists a unique fixed point for the operator $\hat H_n |_{\hat {\mathcal C}(n)\times \mathcal P(X)},$ as under our assumptions the operator $H_n$ has a unique fixed point by Lemma \ref{lemma:e}. 

Let $(Q^{\mathcal C}_0,\mu^{\mathcal C}_0) \in \hat {\mathcal C}(n) \times \mathcal P(X)$ be an initial datum such that $\mu^{\mathcal C}_0$ is obtained from some $\mu^{\mathcal D}_0 \in \mathcal P(X_n)$ as in (\ref{disctocontprob}). Then, there exists $Q^{\mathcal D}_0 \in \mathcal C(n)$ such that $Q^{\mathcal D}_0\circ (T_n,I) = Q^{\mathcal C}_0.$ Recursively, for $n \geq 1,$ define $(Q^{\mathcal D}_n,\mu^{\mathcal D}_n) = H_n(Q^{\mathcal D}_{n-1},\mu^{\mathcal D}_{n-1}).$ Consequently, for the family
$((Q^{\mathcal C}_n,\mu^{\mathcal C}_n))_n$ defined by the recursion $(Q^{\mathcal C}_n,\mu^{\mathcal C}_n) =\hat H_n(Q^{\mathcal C}_{n-1},\mu^{\mathcal C}_{n-1}),$ $n \geq 1,$ we have that $\mu^{\mathcal C}_n$ are obtained by $\mu^{\mathcal D}_n$ via (\ref{disctocontprob}) and $Q^{\mathcal D}_n\circ (T_n,I) = Q^{\mathcal C}_n.$ If $(Q^{\mathcal D}_{\infty},\mu^{\mathcal D}_{\infty}) = H_n(Q^{\mathcal D}_{\infty},\mu^{\mathcal D}_{\infty}),$ then for $Q^{\mathcal C}_{\infty}=Q^{\mathcal D}_{\infty}\circ (T_n,I)$ and $\mu^{\mathcal C}_{\infty} \in \mathcal P(X)$ obtained from $\mu^{\mathcal D}_{\infty}$, the following relations hold:
\[ 
\sup_{(x,a) \in X \times A} | Q^{\mathcal C}_{\infty}(x,a)-Q^{\mathcal C}_{m}(x,a) | = \sup_{(x_n,a) \in X_n \times A} | Q^{\mathcal D}_{\infty} (x_n,a) - Q^{\mathcal D}_{m}(x_n,a)|
\]
and
\[
W_1(\mu^{\mathcal C}_{\infty},\mu^{\mathcal C}_{n}) \leq W_1(\mu^{\mathcal D}_{\infty},\mu^{\mathcal D}_n)
\]
as in Lemma \ref{lemma:c}.
These relations give that the operator $\hat H_{1,n}(Q^{\mathcal C}_m,\mu^{\mathcal C}_m)$ converges to $Q^{\mathcal C}_{\infty}$ in uniform norm and the operator $\hat H_{2,n}(Q^{\mathcal C}_m, \mu^{\mathcal C}_m)$ converges weakly to $\mu^{\mathcal C}_{\infty}.$ Since the fixed point of the operator $\hat H_n |_{\hat {\mathcal C}(n)\times \mathcal P(X)}$ is unique, iterations that start from an element in $\hat {\mathcal C}(n) \times \mathcal P(X)$ will converge to this unique fixed point, so we are done.
\end{proof}

\subsection{Approximation of the Nominal Model Through Quantization}

In this subsection, we will show that through the quantized games, we can find an approximate mean-field equilibrium for the nominal model explicitly via the value-iteration algorithm by using the extended mean-field games.

We will start by showing the continuous convergence of the extended cost function $\hat c_n$ and the extended transition probability $\hat p_n$ to $c$ and $p$ respectively, in their respective topologies, which is required for our robustness result.

\begin{lemma}\label{lemma:d}
For any $(x,\mu) \in X \times \mathcal P(X),$ whenever we have a family $(x_n,\mu^{\mathcal C}_n)_n \subset X \times \mathcal P(X)$ that converges  to $(x,\mu)$ in $X \times \mathcal P(X)$, the following results hold:
\begin{itemize}
\item[i)] The cost function $\hat c_n(x_n,\cdot,\mu^{\mathcal C}_n)$ converges uniformly to $c(x,\cdot,\mu)$ on $A.$
\item[ii)] The weakly continuous family of transition probabilities $\hat p_n(\cdot|x_n,a,\mu^{\mathcal C}_n)$ converges to $p(\cdot | x,a,\mu)$ weakly.
\end{itemize}
\end{lemma}

\begin{proof}
We start by showing that $\hat p_n(\bullet | x,a,\mu^{\mathcal C}_n)$ is a weakly continuous transition probability on $X \times A \times \mathcal P(X).$ For an arbitrary bounded continuous function $g$ on $X$, using the definition of $\hat p_n$ and Assumption \ref{assumption:b}, we obtain 
\begin{align*}
& \sup_{(x,a) \in X \times A \,} \bigg | \int_X g(t)\hat p_n(dt|x,a,\mu_m)-\int _X g(t) \hat p_n(dt|x,a,\mu) \bigg |
\\& =\sup_{(x,a) \in X \times A} \, \bigg | \sum_{t\in X_n} g(t) p(S_{n,i(T_n(t))}|x,a,\mu_m) - \sum_{t\in X_n} g(t)  p(S_{n,i(T_n(t))}|x,a,\mu) \bigg |
\\& \leq \sup_{t'\in X_n}|g(t')| \, \sup_{(x,a) \in X \times A} \, \| p(\cdot|x,a,\mu_m)-p(\cdot|x,a,\mu) \|_{TV} 
\\& \leq \sup_{t'\in X_n}|g(t')|\bigg(K_1 W_1(\mu_m,\mu) \bigg).
\end{align*}
Since $X$ and $A$ are compact it follows that, for each $n$, the transition probability $\hat p_n$ is weakly continuous since continuous convergence can be characterized by convergence on the compact sets.

Note that using the definition of $\hat p_n$ and the bounds in Assumption \ref{assumption:b}, it follows that
\begin{align*}
& \sup_{ (x,a) \in X \times A} \bigg | \int_X g(t)\hat p_n(dt|x,a,\mu^{\mathcal C}_n) - \int_X g(t)p(dt|x,a,\mu) \bigg |
\\&=\sup_{(x,a) \in X \times A}\bigg | \int_{S_{n,i(T_n(x))}}	\sum_{t' \in X_n} g(t')p(S_{n,i(T_n(t'))}|z,a,\mu^{\mathcal C}_n)\nu_{n,i(T_n(x))}(dz) 
\\& \qquad \qquad \qquad - \int_X g(t)p(dt|x,a,\mu) \bigg |
\\& =\sup_{(x,a) \in X \times A} \bigg | \int_{S_{n,i(T_n(x))}}\int_X g\circ T_n(t)p(dt|z,a,\mu^{\mathcal C}_n)\nu_{n,i(T_n(x))}(dz) 
\\& \qquad \qquad \qquad - \int_{S_{n,i(T_n(x))}} \int_X g(t)p(dt|x,a,\mu)\nu_{n,i(T_n(x))}(dz) \bigg |
\\&\leq \sup_{t \in X} |g(t)| \sup_{(x,z)\in X \times S_{n,i(T_n(x))}} \bigg( K_2d_X(z,x) + K_1W_1(\mu^{\mathcal C}_n,\mu)\bigg)+\sup_{t \in X} | g\circ T_n(t)-g(t) |
\\& \leq \sup_{t \in X} |g(t)| \bigg( \frac{K_2}{n} + K_1W_1(\mu^{\mathcal C}_n,\mu) \bigg) + \sup_{t \in X} |g\circ T_n(t)-g(t)|. 
\end{align*}
Since $d_X(T_n(t),t)\leq 1/n$ for any $t \in X,$ we have that $\sup_{t \in X}|g\circ T_n(t)-g(t)| \rightarrow 0$ as $n \rightarrow \infty$ by continuity. The first term above also vanishes as $n \rightarrow \infty $ hence $\hat p_n(\cdot | x,a,\mu^{\mathcal C}_n)$ converges weakly to $p(\cdot|x,a,\mu)$ uniformly in $(x,a).$ 

In a similar fashion, one can show that $\hat c_n(\cdot,\cdot,\mu^{\mathcal C}_n)$ converges uniformly to $c(\cdot,\cdot,\mu)$ uniformly on $X \times A.$ 
\end{proof}

We are in position to state our main results of this section.

\begin{theorem}\label{thrm4}
Suppose $k_1<k_2$ and $k<1$ holds. Under Assumption  \ref{assumption:b}, mean-field equilibria for the quantized games induce mean-field equilibria in the extended setting that converge to that of the nominal game.
\end{theorem}
\begin{proof}
When the conditions $k_1 < k_2$ and $k < 1$ hold, as stated in Theorem \ref{thrm1}, the original game possesses a unique mean-field equilibrium. Utilizing Lemma \ref{lemma:d} and Proposition \ref{prop:d}, we find that the conditions outlined in Theorem \ref{thrm} are satisfied by the extended games derived from the quantized ones. This completes the proof.
\end{proof}

Now, let us try to summarize what we did so far in this section.  Since in case of compact Polish spaces, weak convergence of probability is metrized by the $1$-Wasserstein metric, as a consequence of Theorem \ref{thrm4}, we see that the mean-field equilibrium of the quantized game has an extension that converges to that of the nominal one. Under the assumptions of the Theorem \ref{thrm4}, we can obtain $\epsilon$-near mean-field equilibrium to the original mean-field equilibrium $(\pi_*,\mu_*)$ by picking large enough $n$ such that we can obtain mean-field equilibrium, $(\pi^n,\mu^n)$ for the $n^{th}$ quantized problem which satisfies the relations
\[
 W_1( \pi^n(\cdot|T_n(x_n)),\pi_*(\cdot|x))< \epsilon, 
\]
and
\[
W_1(\mu^n,\mu_*) < \epsilon ,
\]
when $n$ is large enough, where $x_n \rightarrow x$ in $X$. Here the second relation is due to the fact that \( 1 \)-Wasserstein metric metrizes the weak convergence topology when the space \( X \) is compact. We will call $(\pi^n(\cdot|T_n(\cdot)),\mu^n)$ as an $\epsilon$-near mean-field equilibrium for the nominal model.
	
We have the following analogue of Theorem \ref{thrm:a} in case of quantization:
\begin{corollary}
Assume the conditions of Theorem \ref{thrm4} holds. Suppose $n$ is sufficiently large so that $n^{th}$ quantized game induces an $\epsilon/2$-near mean-field equilibrium to that of the nominal model. Let $(\pi_*,\mu_*)$ be the mean-field equilibrium of the nominal model. Then sufficiently large iteration of the value iteration algorithm of the $n^{th}$ quantized game induces an $\epsilon$-near mean-field equilibrium to the nominal model, i.e. if $\pi^{m,n}(\cdot|T_n(\cdot))$ is optimal policy for the extended game under the state measure that is obtained as $m^{th}$ iteration of $n^{
th}$ quantized game, $\mu^{m,n},$ then 
\[
 W_1(\pi^{m,n}(\cdot|T_n(x)),\pi_*(\cdot|x)) \leq \epsilon,
\]
and
\[
 W_1(\mu^{m,n},\mu_*) \leq \epsilon,
\]
whenever $m$ and $n$ are sufficiently large.
\end{corollary}
\begin{proof}
Follows from Theorem \ref{thrm4} and Proposition \ref{prop:d}.
\end{proof}

\section{Robustness of MFGs under General Setting}\label{g-robust}

In this section, we demonstrate that when dealing with MFGs with converging cost functions, transition probabilities, and discounted costs, their mean-field equilibria will converge to those of the limiting case.

Our result in this section represents an improvement over Theorem \ref{thrm} since we no longer require an iteration algorithm to obtain mean-field equilibria for approximate models. In place of that assumption, we assume that the weight function $w_{\max}$ on $X$ is a moment function. This assumption ensures that the set of possible mean-field terms, a subset of $\mathcal P(X)$, is tight.

Since we cannot ascertain in advance that iterations of the $Q$-functions in value iteration algorithms for approximate models and the nominal model will converge to fixed points, the nature of this result implies that obtaining suitably close approximations will be computationally unfeasible  in this setting. This is in contrast to Theorem \ref{thrm}, where we can find approximations of the limiting mean-field equilibrium with a finite number of iterations of the value iteration of the approximating models, as illustrated in Theorem \ref{thrm:a}. Therefore, we can consider the following result as the most general robustness result, but it may lack practical application unless certain conditions ensuring the convergence of the value iteration algorithm for both the approximating and nominal models are imposed.

\begin{theorem}\label{thrm2}
Suppose the following hold:
\begin{itemize}
\item[i)] The state space $X$ is locally compact.
\item[ii)]We have a family of transition probabilities $(p_n)_n$ on $X \times A \times \mathcal P(X)$ that converges to a transition probability $p$ continuously.
\item[iii)]We have a family of cost functions $(c_n)_n$ on $X \times A \times \mathcal P(X)$ that converges to the cost function $c$ continuously.
\item[iv)]We have a family of discounted costs, $(\beta _n)_n,$ that converges to $\beta.$ 
\item[v)] We also have 
\[
\int_X w_{max}(y)p_n(dy|x_n,a_n,\mu _n) \rightarrow \int_X w_{max}(y)p(dy|x,a,\mu)
\]
whenever $(x_n,a_n,\mu_n) \rightarrow (x,a,\mu)$ in $X \times A \times \mathcal P(X).$
\item[vi)] The function $w_{\max}$ is a moment function, and moreover, we have the following bounds: for each $n$
\begin{align*}
 c_n(x,a,\mu) &\leq Mw(x,a),
\\  \int_X w_{\max}(y) p_n(dy|x,a,\mu) &\leq \alpha w(x,a)
\end{align*}
such that $\alpha \beta_n <1$.
\end{itemize}
Moreover, for each $n,$ we define
\[
G_n(x,Q,\mu,a):X \times \mathcal D \times \mathcal P(X) \times A \rightarrow \mathbb R: c_n(x,a,\mu)+\beta_n \int_X Q_{\min}(y)p_n(dy|x,a,\mu)
\]
where $\mathcal D \subset B_{b,w_{\max}}(X \times A)$ is such that for each $Q \in \mathcal D,$ we have $G_n(\cdot,Q,\mu,\cdot) \in \mathcal D.$  Suppose that
the nominal MFG with the components 
$
\bigl( X, A, p, c, \beta \bigr) \nonumber
$
has an unique mean-field equilibrium $(\pi_0,\mu_0)$ and 
\[
G(x,Q,\mu,a):X \times \mathcal D \times \mathcal P(X) \times A \rightarrow \mathbb R: c(x,a,\mu)+\beta \int _X Q_{\min}(y)p(dy|x,a,\mu)
\]
has a unique minimizer for any tuple $(x,Q,\mu).$ For each $n$, let $(\pi_n,\mu_n)_{n}$ be a mean-field equilibrium for the MFG with the components 
\begin{align}
\bigl( X, A, p_n, c_n, \beta_n \bigr)_{n} \nonumber
\end{align}
Then $(\pi_n,\mu_n) \rightarrow (\pi_0,\mu_0)$, where $\pi_n \rightarrow \pi_0$ means that for any $x_n \rightarrow x$, we have $\pi_n(\cdot|x_n) \rightarrow \pi_0(\cdot|x)$ weakly. 
\end{theorem}

\begin{proof}
We first show that $G_n(x_n,Q_n,\mu _n, a)$ converges to $G(x ,Q,\mu , a)$ uniformly in $A$ whenever $(x_n, \mu _n) \rightarrow (x,\mu)$ in $X \times \mathcal P(X)$ and $Q_{n,\min}(x) \rightarrow Q_{\min}(x)$ continuously.

Using triangle inequality we have the following bound
\begin{align*}
&\sup_{a \in A}|G_n(x_n,Q_n,\mu_n,a)-G(x,Q,\mu,a)| 
\\& \leq \sup_{a \in A}|c_n(x_n,a_n,\mu_n)-c(x,a,\mu)|+|\beta_n-\beta|\sup_{a \in A} \bigg|\int_X Q_{\min}(y)p(dy|x,a,\mu) \bigg|
\\& \qquad +\sup_m \beta _m  \sup_{a \in A} \bigg | \int_X Q_{n,\min}(y)p_n(dy|x_n,a,\mu_n) - \int_X Q_{n,\min}(y)p(dy|x_n,a,\mu_n) \bigg |
\\& \qquad + \sup_m \beta _m \sup _{a \in A} \bigg | \int_X Q_{n,\min}(y,a)p(dy|x_n,a,\mu_n) - \int _X Q_{\min}(y)p(dy|x,a,\mu) \bigg |.
\end{align*}
Since $c_n$ converges to $c$ continuously, we have 
$\sup_a|c_n(x_n,a,\mu_n) - c(x,a,\mu)| \rightarrow 0$ as $n \rightarrow \infty .$ Furthermore, as $Q_{\min}(\cdot)$ is bounded by $w_{\max}(\cdot)$ and 
\[
\int_X w_{\max}(y)p(dy|x',a',\mu ') \leq \sup_{a \in A}w(x',a) < \infty,
\] 
we have
\[
|\beta _n - \beta|\sup_{a \in A} \bigg | \int_X Q_{\min}(y)p(dy|x,a,\mu) \bigg | \rightarrow 0.
\] 
Finally, since \( Q_{n,\min}(\cdot) \rightarrow Q_{\min}(\cdot) \) continuously, the inequality
\[
|Q_{n,\min} (\cdot)| \leq | Q_{n,\min}(\cdot)-Q_{\min}(\cdot) | + | Q_{\min}(\cdot)| \leq (\epsilon_n+M)w_{\max}(\cdot),
\]
where \( \epsilon_n \) is a sequence of positive numbers that converge to \( 0 \) and \( M \) is the bound of \( Q_{\min}(\cdot) \) under \( w_{\max} ,\) gives us that the family \( (Q_{n,\min}(\cdot))_n  \) is uniformly bounded by \( w_{\max} \). Thus, since
\[
\int_X w_{\max}(y)p(dy|x_n,a_n,\mu_n) \rightarrow \int_X w_{\max}(y)p(dy|x,a,\mu),
\]
we have
\[
\sup_m \beta _m \sup_a \bigg | \int_X Q_{n,\min}(y)p_n(dy|x_n,a,\mu_n) - \int_X Q_{n,\min}(y)p(dy|x_n,a,\mu_n) \bigg | \rightarrow 0,
\]
and
\[
\sup_m \beta _m \sup _a \bigg | \int_X Q_{n,\min}(y)p(dy|x_n,a,\mu_n) - \int _X Q_{\min}(y)p(dy|x,a,\mu) \bigg | \rightarrow 0 
\]
from \cite[Theorem 3.3]{MR705462}. Thus, a triangle inequality yields that $G_n(x_n,Q_n,a,\mu_n)$ converges to $G(x,Q,a,\mu)$ uniformly on $A.$

Using assumption vi) in the theorem, it is straightforward to prove that the sequence $(\mu_n)_n$ lives in a tight family of probability measures on $X$, and so, has an accumulation point. Let $\mu \in \mathcal P(X)$ be such limit point. Perhaps by switching to a subsequence, we will assume that $(\mu_n)_n$ itself  converges to an accumulation point $\mu$, which will not break the generality as we will show that the limit is state-measure of the mean equilibrium under the limiting game, which is unique by assumption. For each $n,$ it is well known in stochastic control theory that the value iteration operator $G_n(x,Q,\mu_n,a)$ has a unique fixed point $G_n(x,Q_n,\mu_n,a)=Q_n(x,a)$ such that $\supp \pi_n(\cdot|x) \subset \argmin_aQ_n(x,a)$ for all $x$. Furthermore, this fixed point can be obtained as a consequence of mere value iteration for given probability measure $\mu_n$, which is known to exist under the condition $\alpha \beta_n <1.$ 

Now, starting from a common initial $Q$-function, for simplicity say $Q_0=0,$ as a consequence of the argument above $G_n(x_n,Q_0,a,\mu_n)=Q^1_n(x_n,a)$ converges uniformly on $A$ to $G(x,Q_0,a,\mu)=Q^1(x,a)$, and so, by Lemma~\ref{keylemma}, $Q^1_{n,\min} \rightarrow Q^1_{\min}$ continuously. By induction then the very same argument gives that $Q^{k+1}_n=G_n(x_n,Q^k_n,a,\mu _n)$ converges uniformly on $A$ to $Q^{k+1}=G(x,Q^k,a,\mu).$ Let $Q(x,a) = G(x,Q,a,\mu)$ be the fixed point of the operator $G.$ Then we have

\begin{align*}
\sup_a | Q_n(x_n,a) - Q(x,a) | &\leq \sup_a |Q^k_n(x_n,a)-Q_n(x_n,a)| \\& \quad +\sup_a | Q^k_n(x_n,a)-Q^k(x,a) | 
\\& \quad +\sup_a | Q(x,a)-Q^k(x,a)|.
\end{align*}
The first and last terms converge to zero due to the value iteration, and the second term converges to zero from the argument above. Hence, $Q_n$ converges to $Q$ uniformly on $A$ and continuously on $X$. 

By Lemma~\ref{keylemma}, since 
$$\supp \pi_n(\cdot|x) \subset \argmin_a G_n(x_n,Q_n,\mu_n,a)= \argmin_a Q_n(x_n,a)$$ 
for all $x$, for any sequence $\nu_n$ converging to $\nu$ weakly in $\P(X)$, we have $\nu_n \otimes \pi_n \rightarrow \nu \otimes \delta_{f(x,Q,\mu)} = \nu \otimes \pi$ weakly. Hence, if we pick $\nu_n = \delta_{x_n}$ and $\nu = \delta_x$, where $x_n \rightarrow x$ in $X$, we have $\pi_n(\cdot|x_n) \rightarrow \pi(\cdot|x)$ weakly. To complete the proof, we need to prove that $(\pi,\mu)$ is a mean-field equilibrium for the nomimal model as it is unique, i.e., $(\pi,\mu) = (\pi_0,\mu_0)$. But since $\pi(\cdot|x) = \delta_{f(x,Q,\mu)}(\cdot)$ and $G(\cdot,Q,\mu,\cdot) = Q(\cdot,\cdot)$, the policy $\pi$ is optimal for $\mu$. Moreover, for any continuous bounded function $g$, we have
\begin{align*}
&\bigg | \int_X g(y)p(dy|x,f(x,Q,\mu),\mu)\mu(dx) - \int_X g(y)\mu(dy) \bigg | 
\\& \quad \leq \bigg | \int_X g(y)p(dy|x,f(x,Q,\mu),\mu)\mu(dx) 
\\& \qquad \qquad \qquad - \int_{X\times A} g(y)p_n(dy|x,a,\mu _n) \pi_n(da|x)\mu(dx) \bigg |
\\& \quad + \bigg | \int_{X\times A} g(y)p_n(dy|x,a,\mu _n) \pi_n(da|x) \mu(dx) 
\\& \qquad \qquad \qquad - \int_{X\times A} g(y)p_n(dy|x,a,\mu _n) \pi_n \mu _n(dx) \bigg |
\\& \quad + \bigg| \int_{X\times A} g(y)p_n(dy|x,a,\mu _n) \pi_n(da|x) \mu _n(dx) - \int_X g(y)\mu _n(dy) \bigg |
\\& \quad + \bigg| \int_X g(y)\mu _n(dy) - \int_X g(y)\mu (dy) \bigg |
\end{align*}
it follows that 
$$ \bigg | \int_X g(y)p(dy|x,f(x,Q,\mu),\mu)\mu(dx) - \int_X g(y)\mu(dy) \bigg | = 0 $$
since the first and second term converges to \( 0 \) due to Lemma~\ref{keylemma} and \cite[Theorem 3.3]{MR705462}, the third term is identically zero as $\mu_n$'s are state-measure in mean-field equilibria, and the last term converges to \( 0 \) due to the weak convergence. Hence, $(\pi,\mu)$, where $\pi(\cdot|x)=\delta_{\argmin_aQ(x,a)}(\cdot)$, is a mean-field equilibrium for the nominal model. Uniqueness of the mean-field equilibrium completes the proof.
\end{proof}

\begin{remark}
The reason behind rather vague definition of the space $\mathcal D$ is that, depending on the problem at hand one may need to consider handcrafted domain as the space of bounded Borel measurable functions might be too small or too big, depending on the situation.
\end{remark}

\section{Conclusion}

In summary, this paper delved into the robustness of stationary mean-field equilibria when confronted with model uncertainties, particularly in the context of infinite-horizon discounted cost functions. Our approach began by establishing conditions for the convergence of value iteration-based algorithms within mean-field games. Building upon these foundational results, we demonstrated that the mean-field equilibrium obtained through this iterative approach remains robust, even when the system dynamics exhibit misspecifications. Extending the implications of this robustness, we applied these findings to address the finite model approximation problem in mean-field games, revealing that by employing finely tuned state space quantization, we can achieve a remarkably close approximation of the mean-field equilibrium compared to the nominal model.


\end{document}